\documentclass{article}
\usepackage[utf8]{inputenc}
\usepackage{amssymb,amsmath,mathtools}
\usepackage[title]{appendix}
\usepackage{bbm}

\usepackage[margin=1in]{geometry} 

\usepackage{sectsty}

\sectionfont{\fontsize{12}{15}\selectfont}

\DeclareMathAlphabet{\mathpzc}{OT1}{pzc}{m}{it}

\usepackage{xcolor}

\newcommand{\topic}[1]{#1}

\usepackage{mdframed} 

\usepackage{tikz}
\usetikzlibrary{quantikz2}

\DeclareMathOperator*{\argmin}{arg\,min}
\DeclareMathOperator*{\esssup}{ess\,sup}

\DeclareMathAlphabet\mathbfcal{OMS}{cmsy}{b}{n}

\newcommand{\tr}{\operatorname{Tr}}
\newcommand{\norm}[1]{\lVert#1\rVert}

\newcommand\mi{\mkern1.5mu{:}\mkern1.5mu} 

\newcommand{\dm}[1]{ | #1 \rangle \langle #1|}

\newcommand{\score}{\mathpzc{s}}

\newcommand{\mT}{\mathcal{T}}
\newcommand{\mB}{\mathcal{B}}
\newcommand{\mH}{\mathcal{H}} 
\newcommand{\mK}{\mathcal{K}}
\DeclareMathOperator{\cptp}{C} 
\DeclareMathOperator{\Hmin}{H_{min}} 
\DeclareMathOperator{\Dmax}{D_{max}} 
\DeclareMathOperator{\HHH}{H} 
\newcommand{\D}{\operatorname{D}} 

\newcommand{\A}{\mathcal{A}} 
\newcommand{\Ac}{\mathrm{A}} 
\newcommand{\Wc}{\mathbb{W}}
\newcommand{\Vc}{\mathbb{V}}
\newcommand{\B}{\mathcal{B}} 

\newcommand{\mN}{\mathcal{N}}
\newcommand{\mD}{\mathcal{D}}

\newcommand{\I}{\mathbb{I}} 
\newcommand{\E}{\mathbb{E}} 

\newcommand{\X}{\mathcal{X}}
\newcommand{\Y}{\mathcal{Y}}
\newcommand{\mP}{\mathcal{P}} 
\newcommand{\mW}{\mathcal{W}} 
\newcommand{\mV}{\mathcal{V}} 
\newcommand{\mL}{\mathcal{L}}

\newcommand{\C}{\mathbb{C}} 
\newcommand{\R}{\mathbb{R}} 
\newcommand{\U}[1]{\operatorname{U}\left( #1 \right)} 

\newcommand{\Ahat}{\hat{A}}

\usepackage{amsthm}

\makeatletter
\newtheorem*{rep@theorem}{\rep@title}
\newcommand{\newreptheorem}[2]{%
\newenvironment{rep#1}[1]{%
 \def\rep@title{#2 \ref{##1}}%
 \begin{rep@theorem}}%
 {\end{rep@theorem}}}
\makeatother

\theoremstyle{definition}
\newtheorem{definition}{Definition}
\newtheorem{theorem}{Theorem}
\newtheorem*{theorem*}{Theorem}
\newreptheorem{theorem}{Theorem}

\newtheorem{lemma}[theorem]{Lemma}
\newtheorem*{lemma*}{Lemma}

\newtheorem*{corollary*}{Corollary}
\newtheorem{proposition}[theorem]{Proposition}
\newreptheorem{proposition}{Proposition}

\usepackage{authblk}

\usepackage[
sorting=none,
backref=true,
url=false,
isbn=false,
backend=biber,
style=numeric-comp
]{biblatex}
\renewbibmacro{in:}{}
\bibliography{main}
\usepackage{hyperref}
\hypersetup{
     colorlinks   = true,
     citecolor    = blue,
     linkcolor    = blue,
}
\DefineBibliographyStrings{english}{%
  backrefpage = {p},
  backrefpages = {pp},
}
\usepackage{xurl}

\begin{document}
\title{On the coherent extension of some Fano-type learning bounds}
\author[1,2]{Evan Peters}
\affil[1]{Institute for Quantum Computing and Department of Physics, University of Waterloo, Waterloo, Ontario, N2L 3G1, Canada}
\affil[2]{Perimeter Institute for Theoretical Physics, Waterloo, Ontario, N2L 2Y5, Canada}
\date{\today}

\maketitle

\begin{abstract}
    Information theory provides tools to predict the performance of a learning algorithm on a given dataset. For instance, the accuracy of learning an unknown parameter can be upper bounded by reducing the learning task to hypothesis testing for a discrete random variable, with Fano's inequality then stating that a small conditional entropy between a learner's observations and the unknown parameter is necessary for successful estimation. This work first extends this relationship by demonstrating that a small conditional entropy is also sufficient for successful learning, thereby establishing an information-theoretic lower bound on the accuracy of a learner. This connection between information theory and learning suggests that we might similarly apply quantum information theory to characterize learning tasks involving quantum systems. Observing that the fidelity of a finite-dimensional quantum system with a maximally entangled state (the singlet fraction) generalizes the success probability for estimating a discrete random variable, we introduce an entanglement manipulation task for infinite-dimensional quantum systems that similarly generalizes classical learning. We derive information-theoretic bounds for succeeding at this task in terms of the maximal singlet fraction of an appropriate finite-dimensional discretization. As classical learning is recovered as a special case of this task, our analysis suggests a deeper relationship at the interface of learning, entanglement, and information.
\end{abstract}

\section{Introduction}

\topic{A key challenge in applying learning algorithms is knowing how much data is needed to successfully learn from observations. To this end, information theory plays an essential role in understanding the potential performance of a classical machine learning algorithm on some given dataset.} As Shannon theory \cite{shanon1948} is best suited to describing discrete variables, a typical approach to characterize the learning of an unknown parameter involves discretization. For instance, if we wish to characterize a learner's ability to determine an unknown parameter to within $\epsilon$ accuracy, we might divide the parameter space into a set of non-overlapping balls  with radius $\epsilon$, and then consider the problem of determining which ball contains the unknown parameter. Fano's inequality, which characterizes how much mutual information between the learner's observations and the target $\epsilon$-ball is necessary for success, then lower bounds the difficulty of this \textit{hypothesis testing} task \cite{fano_inequality,hasminskii1981}. This technique has been applied to establish lower bounds on the error of learning algorithms involving sparse vectors (e.g. \textit{compressed sensing} \cite{compressedsensing1,compressedsensing2,candes2009power,candes2012exact}) for classical data gathered from both classical \cite{wang2010information,raskutti_minimax_2011,loh_corrupted_2012} and quantum \cite{gross_2010} sources.

\topic{While information theory has been successfully used to lower bound learning error, it is rarely applied to guarantee the performance of a learning algorithm.} In this work, we present such an information-theoretic guarantee: Specifically, we show that in a best-case scenario the performance of an optimal learner improves with decreasing conditional entropy  between the given data and a particular discretization of the unknown parameter space. In other words, we study the \textit{onset of learning} by using classical Shannon theory to show that there exists a sufficient amount of information contained in training data such that some unknown parameter may be learned to $\epsilon$-accuracy.

\topic{Similarly to how Shannon theory characterizes correlations between random variables and provides bounds for hypothesis testing, quantum information theory characterizes correlations in quantum states and provides bounds for certain quantum information processing tasks}. One quantity characterizing the quantum correlations contained in a state is its overlap with a maximally entangled state, maximized with respect to local operations. This quantity, the \textit{maximal singlet fraction}, relates to the state's distillable entanglement \cite{bennet_1996a}, entanglement of formation \cite{bennet_1996b}, negativity \cite{Verstraete_2002}, and usefulness as a resource for teleportation \cite{horodecki_general_1999}. Moreover, the maximal singlet fraction is an information-theoretic quantity  that reduces to the optimal success probability for hypothesis testing as a special case \cite{koenig_operational_2009}. Given this relationship between quantum correlations in a finite-dimensional system and classical correlations among discrete random variables,  might we find an entanglement manipulation task for continuous variable quantum systems that similarly generalizes the notion of classical learning? Moreover, given the established relationship between Shannon theory (for discrete variables) and classical learning (of continuous parameters), might we be able to use the finite-dimensional singlet fraction to characterize the ability to succeed at this continuous variable task?

\topic{In this work we consider such a task, study how it may be understood as a quantum generalization of learning, and demonstrate information-theoretic bounds and guarantees for success.} While prior works have applied statistical learning theory \cite{Huang_2021,caro_generalization_2022,peters2022generalization,caro2022outofdistribution,gilfuster2024understanding} or information theory \cite{PRXQuantum.2.040321, huangprl,caro_information_theoretic_2023} to characterize error in quantum machine learning tasks, these prior tasks still involve extracting \textit{classical} information from quantum (and/or classical) systems. Our work moves beyond the typical setting of quantum machine learning to consider an intrinsically quantum task: A learner receives part of an (infinite-dimensional) bipartite system, and then applies local operations to maximize the fidelity of the final state with a highly entangled state of the bipartite system. We derive upper and lower bounds on the learner's error in this task for a best-case and worst-case initial state, respectively. 

\subsection{Preliminaries}\label{sec:preliminaries}

\subsubsection{Classical and quantum hypothesis testing}

\topic{We first review classical hypothesis testing and its relation to information theory before discussing information-theoretic bounds and guarantees for learning problems.} In classical hypothesis testing (or multiple hypothesis testing) a \textit{target} random variable $A$ taking values $a \in \A$ is sampled according to a discrete probability distribution $p_A: \A \rightarrow \R^+$. Then, \textit{observations} are provided in the form of a random variable $B$ taking values in a finite set $\B$ according to probability $p_{B|A}$ conditioned on the value of $A$. The observations may not be a single data point - for instance, we might have that $\B = \X^n \times \Y^n$ represents a sequence of $n$ many $(x, y)$ pairs of a training set. Given the observations $B=b$ conditioned on sampling $A=a$, the objective is to output an \textit{estimate} $\Ahat(B)$ such that the success probability $p(A|\Ahat) :=\Pr_{p_A}(\Ahat(B) = A)$ is large. Fig.~\ref{fig:1}a provides a schematic representation of classical hypothesis testing.

\topic{We are interested in upper and lower bounds on the error probability for this estimate.} Intuitively, the estimate $\Ahat$ is likely to be wrong if the uncertainty about variable $A$ (given $B$) is large. Fano's inequality formalizes this intuition by lower bounding the failure probability for any estimator $\Ahat = \Ahat(B)$ in terms of the conditional entropy between the inputs of that estimator and the target random variable \cite{fano_inequality}:
\begin{equation}\label{eq:fano}
    \HHH(A|B)  \leq h_2\left(p(A|\Ahat ) \right) + (1 - p(A|\Ahat)) \log(|\A|-1).
\end{equation}
Furthermore, the probability of failure for an optimal estimate also is upper bounded according to 
\begin{equation}\label{eq:minent}
    \max_{\Ahat} p(A|\Ahat) \geq 2^{-\HHH(A|B)},
\end{equation}
where the maximization is over all possible estimators $\Ahat: \B \rightarrow \A$. Bounds similar to (and tighter than) Eq.~\ref{eq:minent} have been derived in various works \cite{kovalevsky1968problem,hellman_probability_1970,feder_relations_1994}, and we will see later that this inequality follows immediately from the definition of  \textit{conditional min-entropy}.

\subsubsection{Maximal singlet fraction}\label{sec:singlet_frac}

\begin{figure}[!t]
    \centering
    
    \begin{tikzpicture}
    \node[scale=1.3] (classical) {
    \begin{quantikz}[wire types={c,c}]
    \lstick{$R$}  & &    &    & \\
   \lstick{$A$} & & \gate{\mathcal{N}} & \gate{\mathcal{D}} & \rstick{$\hat{A}$}
    \end{quantikz}
    };
    \node (blabel) at (0.40,-1.5) [font=\Large]{$B$};
    \draw[->] (blabel) -- (0.40, -0.7);

    \node[scale=1.4, right=of classical] (quantum){
    \begin{quantikz} 
    \makeebit[angle=-30]{$\begin{array}{c}
         R  \\
         A
    \end{array}$}  &  &    & \\
    & \gate{\mathcal{N}} & \gate{\mathcal{D}} & \rstick{$\hat{A}$}
    \end{quantikz}
    };
    \node (blabelq) at (7.6,-1.5) [font=\Large]{$B$};
    \draw[->] (blabelq) -- (7.6, -0.7);
    \node at (-3, 1.5) {(a)};
    \node at (4, 1.5) {(b)};

    \end{tikzpicture}%
    \caption{Schematic comparison of classical multiple hypothesis testing versus maximizing singlet fraction. \textbf{(a)} Letting double lines denote discrete classical random variables and boxes denote stochastic maps on the respective distribution, classical hypothesis testing involves a random variable $A$ taking values $a \in \A$ sampled according to $p_{A} \in \mP(\A)$, which results in observations $B$ according to $p_{B|A} = \mN(p_A)$. The learner samples $\hat{A}(B)$ from $\mD(p_{B|A})$ with the goal of maximizing $\Pr(\hat{A} = A)$ with respect to $\mD$. The reference system plays no role in the classical case, as classical information may be freely copied from $A$ to $R$. \textbf{(b)} In the quantum setting, single lines denote Hilbert spaces with dimension equal to the size of the alphabet of the corresponding random variable (e.g. $\mH_A = \C^{|\A|}$). An initial state $|\psi\rangle_{RA} \in \mH_{RA}$ is evolved by a channel $\mN \in \cptp(\mH_A, \mH_B)$, and the goal of the learner is to maximize the singlet fraction $q(R|B) = |\A|\langle \phi| (\I \otimes \mD) \rho_{RB} |\phi\rangle$ with respect to channels $\mD \in \cptp(\mH_B, \mH_{\Ahat})$, where $|\phi\rangle \in \mH_{R \hat{A}}$ is a maximally entangled state and $\rho_{RB} = (\I \otimes \mN)(\dm{\psi})$. In both cases we consider only the non-asymptotic, or \textit{one-shot} setting: Noisy channel coding \cite{shanon1948} does not apply because there is no source coding for $A$, and Schumacher coding \cite{schumacher_coding} does not apply for analogous reasons.}
    \label{fig:1}
\end{figure}
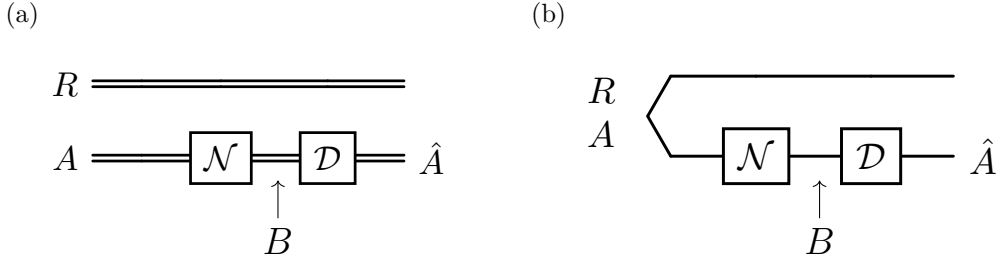%

\topic{We briefly review quantum information theory in finite-dimensional Hilbert spaces.} $\mH$ denotes a finite-dimensional Hilbert space of states with inner product $\langle \psi | \phi\rangle$ denoted using bra-ket notation for $|\psi\rangle, |\phi\rangle \in \mH$ (in the finite-dimensional case, we can choose $\mH = \C^d$ for some $d$). Then, $\D(\mH)$ denotes the \textit{density operators} on $\mH$, i.e. the subset of linear operators $\mL(\mH)$ acting on $\mH$ that are unit trace and positive. We associate a physical system $X$ with a Hilbert space $\mH_X$, and describe a combined system $XY$ in terms of a tensor product of spaces, $\mH_{XY} = \mH_X \otimes \mH_Y$. A linear map $\mN: \mL(\mH) \rightarrow \mL(\mH')$ is positive if it sends positive operators to positive operators, and is completely positive if the map $(\I_n \otimes \mN): (\mL(\mathbb{C}^{n}) \otimes \mathcal{L}( \mH) ) \rightarrow (\mL(\mathbb{C}^{n}) \otimes \mathcal{L}(\mH'))$ is positive for any integer $n$. Then, the set of \textit{channels} from $\mL(\mH)\rightarrow \mL(\mH')$, denoted $\cptp(\mH ,\mH')$,  is the subset completely positive maps which are also trace preserving.

\topic{We now formalize the task of singlet fraction maximization.} Consider a quantum system $RA$ in Figure~\ref{fig:1}b initially containing a pure (possibly entangled) state $|\psi\rangle_{RA} \in \mH_{RA}$, where Hilbert spaces $\mH_R, \mH_A$  have dimension $\dim(\mH_R) = \dim(\mH_A):= d$. Then, the system $A$ evolves via channel $\mathcal{N} \in \cptp(\mH_A, \mH_B)$, and the objective of the learner is to apply an \textit{estimator channel} $\mathcal{D} \in \cptp(\mH_{B}, \mH_{\Ahat})$ (with $\dim(\mH_{\Ahat})= \dim(\mH_A)$) in order to maximize the overlap of the final state with a maximally entangled state. The \textit{singlet fraction} is defined as this overlap, up to the dimensionality of the system:
\begin{align}\label{eq:qra}
    q(R|B)_{\rho} &:= d \, \langle \phi|(\I_R \otimes \mathcal{D})(\rho_{RB}) | \phi\rangle_{R\Ahat},
\end{align}
where 
\begin{align}
    \rho_{RB} &= (\I_R \otimes \mN)(\dm{\psi}_{RA} ) \label{eq:rhoRB}
    \\ |\phi\rangle_{RA} &= \frac{1}{|\A|}\sum_{a \in \A} |a_R\rangle \otimes |a_A\rangle
\end{align}
are the state on $RB$ after the evolution $\mathcal{N}$ and a canonical maximally entangled state respectively. Unlike the classical quantity $p(A|\Ahat)$, $q(R|B)_\rho$ takes values in the interval $[0, d]$. The singlet fraction can be used to characterize the ability of $\mD$ to recover quantum information from noise: If the input state were a singlet, $|\psi\rangle_{RA} = |\phi\rangle_{RA}$, then by the Choi-Jamiołkowski  correspondence \cite{jamiolkowski_linear_1972,choi_completely_1975} the maximal $q(R|B)_{\rho}=d$ is achieved if and only if $\mD \circ \mN = \I$, and if $\mD \circ \mN$ is close to $\I$ (in diamond norm) then $q(R|B)_{\rho}$ will be close to maximal \cite{4920}. However, moving forward we will not require the input state $|\psi\rangle_{RA}$ to be maximally entangled.

\topic{We can express the maximal value of $q(R|B)_{\rho}$ in terms of an information-theoretic quantity called the \textit{conditional min-entropy}.} Given a finite-dimensional bipartite quantum system $\mH_A \otimes \mH_B$, the conditional min-entropy is defined with respect to a state $\rho \in \D( \mH_{AB})$ as:
\begin{align}
    \Hmin(A|B)_\rho &= - \inf_{\sigma \in \D(\mH_B)} \Dmax(\rho || \I \otimes \sigma)
\end{align}
where $\Dmax(\rho || Q) = \inf \{\lambda \in \R: \rho \leq 2^\lambda Q\} $ for a positive operator $Q \geq 0$  (e.g. \cite{watrous2018advanced}). For any state, it holds that $\Hmin(A|B)_\rho~\leq \HHH(A|B)_\rho$. $\Hmin$ has an operational interpretation for both classical distributions and quantum states given by \cite{koenig_operational_2009}: By choosing $\mH_A = \C^{|\mathcal{A}|}$ and $\mH_B = \C^{|\mathcal{B}|}$ and defining a diagonal state $p$ with elements $p_{AB}(a,b)$ on its diagonal, we have
\begin{align}
    2^{-\Hmin(A|B)_p} &=  \max_{\Ahat: \B \rightarrow \A} p(A|\Ahat (B)) \label{eq:minent_1}\\
    2^{-\Hmin(A|B)_\rho} &= \max_{\mathcal{D} \in \cptp(\mH_B, \mH_{\Ahat})} q(R|B)_{\rho} .
\end{align}
Noting that $q(R|B) \leq 1$ for separable states (and therefore classical distributions represented by $p$) \cite[Example 6.16]{watrous2018theory}, we recover the statement of Eq.~\ref{eq:minent} about optimal success probability for classical hypothesis testing from Eq.~\ref{eq:minent_1}

\topic{Just as the error of classical hypothesis testing is both upper and lower bounded in terms of the conditional Shannon entropy, the maximal singlet fraction may be bounded in terms of quantum entropies.} The lower bound on $q(R|B)_{\rho}$ for the optimal estimator becomes\footnote{While the guarantee for classical hypothesis testing is a special case of the conditional min-entropy in the same way that linear programming is a special case of semi-definite programming, it is interesting that the former was derived in Ref.~\cite{feder_relations_1994} by different means entirely.} \cite{koenig_operational_2009}:
\begin{equation}
    \max_{\mathcal{D} \in \cptp(\mH_B, \mH_{\Ahat})} q(R|B)_{\rho} = 2^{-\Hmin(R|B)_\rho} \geq 2^{-\HHH(R|B)_\rho}.
\end{equation}
Furthermore, we may upper bound $q(R|B)_{\rho}$ using an inequality that is superficially similar to Fano's inequality. The following represents a family of such inequalities: 
\begin{proposition}[Quantum Fano inequality \cite{quantumfano}]\label{prop:qfano} 
    For any pure state $|\psi\rangle \in \mH_{RA}$ and any $\rho \in \D(\mH_{RB})$ with $\dim(B)=\dim(A) = d$, define $p = \langle \psi | \rho | \psi \rangle$. Then,
    \begin{equation}
        \HHH(RB)_\rho \leq h_2(p) + (1-p) \log(d^2 - 1)
    \end{equation}
\end{proposition}
This inequality was originally derived in Ref.~\cite{quantumfano} to relate the \textit{entropy exchange} (or entropy of environment plus reference system) to the \textit{entanglement fidelity} (overlap of the noisy state $\rho = (\I_R \otimes \mN)(\dm{\psi})$ with initial state $|\psi\rangle$). While that choice leads to a relatively tight inequality that is relevant for, e.g. quantum error correction, it does not provide a bound on singlet fraction. Consider instead the following instantiation of Proposition~\ref{prop:qfano}:
\begin{equation}
    \HHH(RB)_\rho \leq h_2\left(\frac{q(R|B)_{\rho}}{d}\right) + \left(1-\frac{q(R|B)_{\rho}}{d}\right) \log(d^2 - 1)
\end{equation}
where $\rho=\rho_{RB}$ is defined in Eq.~\ref{eq:rhoRB}. In this setting, the above inequality provides an analogous bound for maximizing singlet fraction as Fano's inequality provides for classical hypothesis testing. $\HHH(RB)_\rho$ appears here rather than $\HHH(R|B)_\rho$ since this bound is derived from a spectral decomposition of the bipartite system $RB$, which does not relate to the conditional entropy in any general sense.

\topic{Between the tasks of classical hypothesis testing and maximizing singlet fraction is \textit{quantum hypothesis testing}, which describes the scenario in which $A$ represents a discrete random variable and $B$ represents a quantum system (e.g. \cite{PhysRevLett.98.160501,audenaert_asymptotic_2008})}. In this scenario, conditional on $A=a$ the learner receives a quantum state $\sigma_B^a \in \D(\mH_B)$ and the learner's task is to determine $a$ by performing measurements on the quantum system. This scenario can be characterized via a joint classical-quantum system $AB$ prepared in a state $\rho= \sum_a p_A(a)\dm{a} \otimes \sigma_B^a$, such that $\Hmin(A|B)_\rho$ determines the learner's optimal probability of identifying $A$ using measurements on system $B$ \cite{koenig_operational_2009}.

\subsubsection{The reduction from singlet fraction to classical hypothesis testing}\label{sec:cht_reduction}
\topic{We briefly review how maximizing singlet fraction reduces to the objective of classical hypothesis testing.} Our main tool will be a \textit{completely dephasing map} $\Delta \in \cptp(\mH)$ defined with respect to a fixed basis $\{|i\rangle\} \subset \mH$ by its action $\Delta(\rho) = \sum_i \dm{i} \langle i| \rho |i\rangle$. Applying this map to both subsystems of the state $\rho_{RB}$ reduces the singlet fraction to the probability of success for a classical estimator:
\begin{align}\label{eq:finite_dephase}
    \langle \phi_{R\Ahat}, (\I \otimes \mathcal{D} ) \circ ( \Delta \otimes \Delta) (\rho_{RB})   \rangle &= \langle  (\Delta \otimes \I)(\phi_{R\Ahat}), (\I \otimes \mathcal{D}\circ   \Delta)(\rho_{RB})\rangle 
    \\&= \frac{1}{d} \sum_{r, b} p_{RB}(r, b) \langle r| \mathcal{D} (\dm{b})|r\rangle
\end{align}
where $p_{RB}(r, b) := \langle rb | \rho | rb\rangle$ and we have used the self-adjointedness of $\Delta$, $\langle X, \Delta(Y)\rangle = \langle \Delta (X), Y\rangle$. Fixing the values $r, b$ and maximizing the final system over estimator channels gives, 
\begin{equation}
    \max_{\mathcal{D} \in \cptp(\mH_B, \mH_{\Ahat})} \langle r| \mathcal{D} (\dm{b})|r\rangle = \max_{\mathcal{D} \in \cptp(\mH_B, \mH_{\Ahat})} \langle r| \Delta \circ \mathcal{D} \circ\Delta(\dm{b})|r\rangle  := \max_{\Ahat: \B\rightarrow \A} \Pr(\Ahat(B) = r|B=b).
\end{equation}
Writing $\Pr(\Ahat(B) = r|B=b) := p_{\Ahat|B}(r| b)$ describing the conditional probability of a measurement outcome $r$, we recover 
\begin{align}
    \sum_{r,b}  p_{RB}(r, b) p_{\Ahat|B}(r| b) &= \sum_{r, b} p_{RB}(r, b) p_{\Ahat | RB}(r | r, b) 
    = \sum_{r} p_{R \Ahat}(r, r) = \Pr_{p_R} (\hat{A}= R),
\end{align}
where we have used conditional independence of $\hat{A}$ and $R$ given $B$. Combining our results, 
\begin{equation}
    \max_{\mathcal{D} \in \cptp(\mH_B, \mH_{\Ahat})} d \,\langle \phi_{R\Ahat}, (\I \otimes \mathcal{D} ) \circ ( \Delta \otimes \Delta) (\rho_{RB})   \rangle = \max_{\Ahat: \B\rightarrow \A} \Pr_{p_R} (\hat{A}= R).
\end{equation}
\topic{We see that a singlet fraction maximization problem is naturally associated with, or \textit{reduces to}, a classical hypothesis testing problem via application of a completely dephasing channel.} Moving forward, when we introduce our quantum information processing task that reduces to classical learning it will be this kind of reduction that we seek.

\section{Bounds and guarantees for classical learning}\label{sec:classical}

\topic{An essential task in the theory of machine learning is to understand how well an algorithm might perform on some given learning problem. We now formalize the learning problem that we will study, before connecting classical learning to a more general quantum information processing task}.  We begin with preliminaries: Given a set $\X$, we define the set of probability densities $\mP(\X)$. Then, for a probability density $p\in \mP(\X)$ we let $\E_p[f(x)]$ denote the expected value of $f$ with respect to $p$ and define the \textit{indicator function} $\I\{x \in S\}$ which is equal to one if $x \in S\subseteq \X$ and zero otherwise. We write the probability of some event $\omega \subseteq \X$ with respect to $p$ as $\Pr_p (\omega):= \E_p[\I\{x \in \omega \}]$.

To define our learning problem, we assume that nature has chosen some fixed \textit{target parameter} $\alpha \in \Ac$, where $\Ac$ is a bounded subset of some metric space $X$ with metric $d: X \times X \rightarrow \R^+$. We choose $X$ to be a metric space so that we can judge estimates for $\alpha$ in terms of the distance $d$ (for simplicity, the reader may consider $X$ to be a normed vector space, e.g. $d(x,y) = \norm{x - y}_2$ for $x,y \in \mathrm{A} \subset \R^p$). We then assume that \textit{observations} are described by a random variable $B$ taking values in a set $\B$ according to probability $p_{B|A}$ conditional on the choice of $\alpha$. We need not specify additional properties of $\B$ for now, but as mentioned before the observations $B$ might represent a training dataset for supervised machine learning, so that $\B = \X^n \times \Y^n$ is a sequence of $n$ pairs of data points from $\X$ and labels from $\Y$ and $B=\{(X_1, Y_1), \dots, (X_n, Y_n)\}$. Equivalently, the distribution $p_{B|A}$ may be considered as the output of a stochastic map $\mN$ acting on $p_A$. Given observations $B$, the learner tries to estimate $\alpha$ using an \textit{estimator} $\hat{\alpha}: \B \rightarrow \Ac$. The learner's estimate is scored according to its distance from the true $\alpha$. We will specify a scoring function later, but the essence is that the learner tries to minimize $d(\hat{\alpha}(B), \alpha)$.

\topic{The problem we study here differs from some other common models of learning.} This learning task mostly resembles parameter estimation, and differs from \textit{prediction} tasks common in machine learning where the learner's task is to construct a model that approximates an unknown target function with low error, with respect to the distribution of data. Probably Approximately Correct (PAC) learning is an example of a formal learning model for prediction tasks \cite{valiant1984theory}. In PAC learning, an unknown function $f_\alpha: \X \rightarrow \Y$ is prepared and the goal of the learner is to output another function $h:\X\rightarrow \Y$ such that the error $\Pr_{x \sim p_X}(h(x) \neq f_\alpha(x))$ is small, with high probability. For instance, in a linear least squares problem the learner receives observations $B = \{(x_1, f_\alpha(x_1) + \epsilon_1), \dots, (x_n, f_\alpha (x_n) + \epsilon_n)\}$ with $x_i \sim p_X(x)$ and zero-mean error $\epsilon_i$, and attempts to approximate  $f_\alpha(x) := \langle \alpha, x\rangle$ with small error with respect to $p_X$. Instead, in our learning setting the learner's goal is to output $\hat{\alpha}$ close to $\alpha$, a task which may generally be harder than PAC learning of $f_\alpha$.

\subsection{Classical learning lower bound via Fano method}

\topic{We first review a well-known lower bound on the error of a learning task, based on Fano's inequality.} This bound works by considering a discretized version of the learning task presented above, such that learning $\alpha$ within $\epsilon$-error is harder than some hypothesis testing problem, which in turn has a success probability upper bounded by Fano's inequality. One kind of discretization that achieves this is an \textit{$\epsilon$-packing}:

\begin{definition}{$\epsilon$-packing:}
    $\mathcal{V}:= \mathcal{V}(\epsilon, \Ac, d)\subset \Ac$ is an $\epsilon$-packing of $\Ac$ in $d$ if for every $\alpha_v, \alpha_{v'} \in \mathcal{V}$, $d(\alpha_v, \alpha_{v'}) \geq \epsilon$.
\end{definition}
Since $\Ac$ is bounded, $|\mathcal{V}| < \infty$. We will assess the error of the learner's estimate with some non-decreasing loss function $\ell: \R^+ \rightarrow \R$, where smaller $\ell$ means a better estimate. Continuing with the previous example with $\Ac \subset \R^p$, choosing $d(\alpha, \alpha') := \norm{\alpha - \alpha'}_2$ and $\ell(x) = x^2$ corresponds minimizing squared error. Then, achieving optimal performance on the learning task turns out to be no harder than performing classical hypothesis testing for the nearest point to $\alpha$ in a $2\epsilon$-packing of $\Ac$, as the following bound from the literature shows:
\begin{proposition}[Classical lower bound, \cite{hasminskii1981}]\label{prop:minimax}
Fix $\epsilon>0$ and let $\mathcal{V}:= \mathcal{V}(2\epsilon, \Ac, d)$ be an $\epsilon$-packing of $\Ac$ in $d$. Define a random variable $V$ uniformly distributed on $\{1, 2, \dots, |\mathcal{V}|\}$, and $\hat{V}: \B \rightarrow \{1, \dots, |\mathcal{V}|\}$ as an estimator for the element of $\mathcal{V}$ closest to $\alpha$ given observations $B$:
\begin{equation}\label{eq:Vhat}
    \hat{V}(r) = \{v: \argmin_{\alpha_{v'} \in \mV} d(\alpha_{v'}, r) = \alpha_v\}.
\end{equation}
Then the following \textit{minimax} lower bound holds:
\begin{equation}
    \min_{\hat{\alpha} } \max_{\alpha \in \Ac } \underset{{p_{B|\alpha}}}{\E} \left[ \ell(d(\hat{\alpha}(B), \alpha))\right] \geq \ell(\epsilon)  \left(  \frac{\HHH(V|B) - 1}{\log|\mathcal{V}|} \right).
\end{equation}
For completeness, we provide a standard proof of this statement in Appendix~\ref{app:minimax}.
\end{proposition}
This bounds the accuracy of an optimal estimator outputting the hardest-to-learn element in $\Ac$. Thus, the bound describes the best learning performance (with respect to learning algorithms) in a worst-case scenario. Some intuition behind this bound: we can tune the bound's sensitivity to error by choosing a larger or smaller $\epsilon$. If $\epsilon$ is made very large, then $|\mathcal{V}|=1$ is the largest packing of $\Ac$, in which case the corresponding hypothesis testing task is trivial but the error of the estimate $\hat{\alpha}$ might be intolerably large. On the other hand, a small $\epsilon$ results in a large $|\mathcal{V}|$ and the error lower bound becomes vacuous. In general, this bound states that a mutual information of roughly $I(V\mi B) \approx \HHH(V)$ bits between the observations $B$ and the index of the $\epsilon$-ball containing $\alpha$ is necessary for \textit{any} estimator to succeed, but does not provide any statement as to whether such a condition is sufficient for success. Ref.~\cite{meyer2023quantum:QIP:24} derived a similar bound to Proposition~\ref{prop:minimax} for when observations $B$ consist of quantum states, by reducing the corresponding learning problem to quantum hypothesis testing on an appropriately defined discretization.

\subsection{Classical learning upper bound}

\topic{In the previous section we saw how a learning problem may be discretized into a hypothesis testing problem for determining a discrete random variable $V$ such that about $\HHH(V)$ bits of mutual information $I(V\mi B)$ between the learner's observations and $V$ are necessary to succeed.} Here we will derive an amount of mutual information that is sufficient for an optimal learner to learn \textit{something} about the desired variable. The guarantee works by discretizing the learning task differently than before, such that learning within $\epsilon$-error is harder than a corresponding classical hypothesis testing problem, which in turn has a success probability lower bounded according to Eq.~\ref{eq:minent}. In this case, the relevant discretization involves an \textit{$\epsilon$-net}:
\begin{definition}{$\epsilon$-net:}
    $\mW:= \mathcal{W}(\epsilon, \Ac, d) \subset \Ac$ is an $\epsilon$-net of $\Ac$ in $d$ if for every $\alpha \in \Ac$, there exists $\alpha_w \in \mathcal{W}$ such that $d(\alpha, \alpha_w) \leq \epsilon$.
\end{definition}
To make our discretization arguments using the $\epsilon$-net, we will introduce the notion of an $\epsilon$-\textit{covering partition}, which divides the subset $\Ac$ into regions consisting of all points that are closest to any given $\alpha_w \in \mW$:
\begin{definition}[$\epsilon$-covering partition]\label{def:ecp}
    Let $\Ac$ be a bounded subset of a metric space and $\mW:= \mW(\epsilon, \Ac, d)$ be an $\epsilon$-net of $\Ac$ in $d$. The set of closest points in $\mW$ to $r$ is 
    \begin{equation}
        \hat{W}(r) \:= \{w: \argmin_{\alpha_{w'} \in \mW} d(\alpha_{w'}, r) = \alpha_w\}
    \end{equation}
    Then we define an $\epsilon$-covering partition of $\Ac$ with respect to $\mW$ as
    \begin{align}
        C(\Ac, \mW) &= \{R_w: w=1, \dots, |\mW|\},
        \\ \text{ where }R_w &= \bigl\{r \in \Ac: \hat{W}(r) \ni w \bigr\}.
    \end{align}
\end{definition}
$C(\Ac, \mW)$ forms a partition of $\Ac$, since for all $R_w, R_{w'} \in C(\Ac, \mW)$, $R_w \cap R_{w'} = 0$ for $w \neq w'$ and $\cup_{w =1\dots|\mW|} R_w = \Ac$ with ties broken arbitrarily, and the topology (e.g. openness) of any given subset $R_w$ may also be assigned arbitrarily. An important property of $C(\Ac, \mW)$ is that any estimator that outputs $w$ for all $\alpha \in R_w$ is guaranteed to be $\epsilon$-accurate. In this setting, the learner attempts to maximize a non-increasing \textit{score} $\score: \R^+ \rightarrow \R^+$ that potentially depends nonlinearly on the distance $d$. For example, choosing $\score(t) = c - t$ (with $c$ being some large-enough constant like the diameter of $\Ac$) means that maximizing $\score (d(x, y))$ corresponds to minimizing the distance between elements $x$ and $y$ (larger $\score$ means a better estimator). We may now state the guarantee for the above learning problem:

\begin{proposition}[Classical learning upper bound]\label{prop:maximax}
Let $\mW:= \mW(\epsilon, \Ac, d)$ be an $\epsilon$-net of $\Ac$ in $d$ and define a random variable  $W$ uniformly distributed on $\mathcal{W}(\epsilon, \Ac, d)$, and let $\score:\R^+ \rightarrow \R^+$ be a non-increasing strictly positive function used to score the accuracy of the estimator. Then, as shown in Appendix~\ref{app:maximax} we have,
\begin{equation}
       \max_{\hat{\alpha}} \max_{\alpha \in \mathrm{A}} \log  \E_{p_{B|\alpha}} [\score(d(\hat{\alpha}(B), \alpha))] \geq \log\score(\epsilon) -\HHH(W|B).
\end{equation}
\end{proposition}
In contrast with Proposition~\ref{prop:minimax}, this bound describes the optimal performance on a learning algorithm (with respect to estimators $\hat{\alpha}$) in a best-case scenario, i.e. the easiest-to-learn $\alpha$. We might describe this as the onset of learning: Whenever this bound is nontrivial, then it must be true that \textit{something} may be learned given the right observations $B$. The logarithm on the LHS of Proposition~\ref{prop:maximax} (compared to Proposition~\ref{prop:minimax}) appears because Eq.~\ref{eq:minent} is loose, compared to Fano's inequality which is sharp \cite{cover1999elements}. This is easily rectified by substituting a tight bound (Ref.~\cite{feder_relations_1994}, Theorem 1) at the expense of simplicity.

\topic{This classical upper bound provides insight into the relationship between learning and classical correlations.} In Proposition~\ref{prop:maximax}, a learner is successful if they are able to output $\hat{\alpha}$ within $\epsilon$ of the target variable $\alpha$ when provided observations $B|\alpha$. The maximum probability that this is satisfied (with respect to learning algorithms $\hat{\alpha}$) is,
\begin{equation}\label{eq:class_wit}
    \max_{\hat{\alpha}}\Pr_{p_{\hat{\alpha}|\alpha}} \left( d(\hat{\alpha}, \alpha) \leq \epsilon \right) := \max_{\hat{\alpha}}\underset{p_{\hat{\alpha}|\alpha}}{\E}  \bigl[ \I \{d(\hat{\alpha}, \alpha) \leq \epsilon \}\bigr] =  \max_{\hat{\alpha}}\underset{p_{\hat{\alpha}|\alpha}}{\E} \bigl[ \I \{\hat{\alpha} \in \mathbb{B}_\epsilon(\alpha) \}\bigr], 
\end{equation}
where $\mathbb{B}_\epsilon(\alpha) = \{\alpha' \in \Ac: d(\alpha, \alpha') \leq \epsilon\}$ is the intersection of the $\epsilon$-ball centered at $\alpha$ with $\Ac$. The final expression can be interpreted as the expected value of a witness for classical correlation between two continuous variables $\alpha$, $\hat{\alpha}$, i.e. the classical learner is successful if their estimate $\hat{\alpha}$ results in a large expected value of the operator $\I \{\hat{\alpha} \in \mathbb{B}_\epsilon(\alpha) \}$. 

\section{A quantum generalization of learning}

\topic{A primary challenge for quantum machine learning is to identify tasks for which quantum information processing might provide an advantage over classical methods.} Accordingly, a large body of work has been devoted to proposing and studying settings where quantum computers could offer an advantage in learning from classical data \cite{huang2021power,kubler2021inductive, Liu_2021,Gyurik2022towardsquantum,Huang_2022,gyurik2023exponential} and for characterizing quantum states and processes \cite{Khatri2019quantumassisted,nfl_2022,Caro_2023v,caro2023learning,huang_arbitrary,jerbi2023power}. Some of these tasks involve data that are provided in superposition and may be manipulated coherently by the learner, such as the Harrow–Hassidim–Lloyd algorithm for linear systems of equations \cite{hhl_2009}. In particular, the model of quantum PAC learning \cite{bj99} generalizes PAC learning by providing a learner access to observations (training data) in coherent superposition generated via an unknown concept (target function), but otherwise the learner's goal is still to approximate the unknown concept to high accuracy with respect to the distribution of observations. Quantum learners in the quantum PAC setting have been shown to provide advantages in sample complexity for specific learning tasks \cite{qpac1,qpac2}, though this advantage can disappear under more general conditions \cite{arunachalam2018optimal}. 

\topic{In order to reason about the relationship between entanglement and learning, we first consider an example of a formal learning scenario similar to the classical learning task introduced above.} The framework of \textit{exact learning} involves an unknown concept described by a discrete variable, and the learner must identify this variable \textit{exactly} with high probability given a set of observations \cite{exact_learning}. The generalization to quantum exact learning \cite{ambainis_quantum_2004,kothari2014}, in which the learner has access to observations described by a quantum state, will help us reason about how we might generalize our classical learning task involving continuous variables to a quantum task involving infinite dimensional states.

\subsection{An example: Beyond quantum exact learning}\label{sec:qexact}

\topic{Fundamentally, despite providing a learner with quantum data, the quantum PAC and quantum exact learning settings described above still involve extracting classical information about an unknown target function.} This limitation suggests a more general setting in which the target functions themselves are treated as \textit{quantum data} that exist in coherent superposition. To demonstrate this concept in a limited (i.e. finite dimensional) setting, we will consider a generalization of classical (and quantum) exact learning that mirrors the generalization of classical learning that we will soon introduce.

\topic{We review the setting of classical exact learning in more detail.} In this setting, without loss of generality, a concept (boolean function with bitstring inputs) $c_A$ is sampled from a concept class $\mathcal{C} \subseteq \{c_a: \{0,1\}^n \rightarrow \{0, 1\} \}$ with some probability $p_A(a)$. Once $A=a$ is fixed, a learner is given oracle access for evaluating $c_a(x)$ for input bitstrings $x\in\{0,1\}^n$, and tries to determine $c_a$ with high probability with resepct to a distribution of concepts $c$. Specifically, a successful exact learner outputs a hypothesis $h: \{0,1\}^n \rightarrow \{0,1\}$ such that $h(x) = c(x)$ for all $x \in\{0,1\}^n$ with probability at least 2/3. Considering Figure~\ref{fig:1}a, this scenario is modelled by introducing a random variable $A$ taking values $a \in \A:= \{1, \dots, |\mathcal{C}|\}$ with probability $p_A(a)$, and then describing the learner's $m$ queries via a random variable $B$ sampled from $p_{B|A} = \mathcal{N}(p_A)$ conditioned on $A$:
\begin{equation}
    B = \{(x_1, c_A(x_1)), \dots, (x_m, c_A(x_m))\},
\end{equation}
where each $x_i$ is sampled from some distribution $p_X$, which can be completely described by $\mN$. For consistency of notation, we describe the joint distribution of $(A, B)$ as a bipartite classical probability distribution (diagonal density operator) of the form
\begin{equation}\label{eq:class_exact}
    p_{AB} = \sum_{a \in \A} p_A(a) \dm{a}_A \otimes p_B^a
\end{equation}
where $p_B^a = \sum_{b \in \B} p_{B|A}(b|a) \dm{b}_B$ is a conditional distribution. In this case, the learner's strategy to identify $c_a$ given $\{(x_i, c_a(x_i))\}_{i=1}^m$ is equivalent to finding an estimator $\Ahat = \Ahat(B)$ that maximizes $\Pr_{p_A}(\Ahat=A)$. Then, for any fixed choice $A=a$ of target function, sampling the random variable $B$ corresponds to preparing a size-$m$ classical dataset of $(x, c_a(x))$ pairs, and we may view exact learning as a kind of classical hypothesis testing for the random variable $A$.

\topic{Quantum exact learning generalizes exact learning by providing samples in superposition via an oracle $|x, 0\rangle \rightarrow |x, c_A(x)\rangle$, while the concept $c_A$ is still described by a (classical) random variable.} Retaining all of the notation from classical exact learning, this scenario is characterized by a channel $\mN \in \cptp(\mH_A, \mH_B)$ composed of oracles and classical sampling from $p_X$ that prepares a classical-quantum state of the form
\begin{equation}\label{eq:quantum_exact}
    \sum_{a \in \A} p_A(a) \dm{a}_A \otimes \dm{\psi_{x}^a}_B
\end{equation}
where $|\psi_{x}^a\rangle = \sum_{i=1}^m \sqrt{p_X(x_i)} |x_i, c_a(x_i)\rangle \in \mH_B$ is a coherent superposition of examples $x_i$ along with the concept $c_a$ evaluated for that input. The goal of a quantum exact learner is to identify $c_a$, which they can accomplish by learning some channel $\mD\in \cptp(\mH_B, \mH_{\Ahat})$ acting on $\mH_B$ that outputs an estimate $\Ahat$ maximizing $\Pr_{p_A}(\Ahat=A)$. From the arguments of Sec.~\ref{sec:cht_reduction}, for each exact learning problem defined by priors $p_A$ and conditionals $p_B^a$ in Eq.~\ref{eq:class_exact}, there is a quantum exact learning problem that is at least as difficult, which may be constructed by choosing states $\dm{\psi_{x}^a}_B$ such that $\Delta(\dm{\psi_{x}^a}_B) = p_B^a$.

\topic{The generalization from classical to quantum exact learning immediately suggests a further generalization of quantum exact learning in which the concepts $c_A$ are themselves in superposition.} This task should be at least as hard as quantum exact learning, and should reduce to quantum exact learning via application of a completely dephasing channel $\Delta$. Referring to Fig.~\ref{fig:1}b, suppose that $RA$ initially contains a state $|\psi_{RA}\rangle := \sum_{a \in \A} \sqrt{p_A(a)} |a\rangle_R \otimes |a\rangle_A$. Constructing $\mN$ from oracles and classical sampling as before, we prepare the state
\begin{equation}
    \sum_{a \in \A} \sqrt{p_A(a)} \ket{a}_R \otimes \ket{\psi_{x}^a}_B.
\end{equation}
In this case, what should the goal be for a learner who has the ability to apply local operations to the system $B$? Based on the discussion of Sec.~\ref{sec:preliminaries}, the most natural choice is for the learner to learn a channel $\mD \in \cptp(\mH_B, \mH_{\Ahat})$ that maximizes the overlap of the final state $(\I \otimes \mD)(\sigma_{RB})$ with a maximally entangled state in $R\Ahat$. Based on the arguments of Sec.~\ref{sec:cht_reduction}, this choice to maximize the singlet fraction $q(R|B)_{\rho}$ results in a task that is at least as hard as quantum exact learning, and recovers the setting of quantum exact learning by application of the completely dephasing map on register $A$.

\subsection{The entanglement fraction task}

\topic{We are now ready to characterize a task which generalizes the setting of classical learning.} Recall that in Sec.~\ref{sec:preliminaries} we saw how classical hypothesis testing (involving discrete random variables) can be understood as a special case the more general task of achieving the maximal singlet fraction (involving finite dimensional Hilbert spaces and reference systems). In Sec.~\ref{sec:classical}, we saw that hypothesis testing provides bounds (via discretization arguments) for the error of classical learning. In Sec.~\ref{sec:qexact} we observed how exact learning (involving classical or quantum information) can be recovered as a special case of a task involving the maximal singlet fraction. We will now combine all of these ideas to introduce a quantum information processing task that generalizes the notion of classical learning, and furthermore we will show that the error of an optimal learner for this task may be bounded using tools from quantum information theory for finite-dimensional systems. 

\topic{We briefly review quantum information theory in infinite-dimensional Hilbert spaces, overloading notation from the previous sections in the process.} $\mH$ denotes a separable Hilbert space of states, with inner product denoted $\langle \Psi | \Phi\rangle$ for $|\Psi\rangle, |\Phi\rangle \in \mH$. $\mB(\mH)$ denotes bounded linear operators on $\mH$, and $\mT(\mH)\subseteq \mB(\mH)$ denotes the trace class, self-adjoint operators on $\mH$. The subset $\D(\mH) \subset \mT(\mH)$ of unit trace operators denotes the \textit{density operators}. We will sometimes denote a pure state in density operator form as $\dm{\Psi} := \Psi \in \D(\mH)$ when there is no risk of confusion, and we will tend to use uppercase (e.g. $\Phi, \Sigma$) for states of an infinite dimensional Hilbert space and lowercase (e.g. $\phi$, $\rho$, $\sigma$) for states of a finite dimensional space. We use bra-ket notation to denote both states in $\mH$ as well as (unbounded) linear functionals, e.g. sets of operators $\{|\alpha\rangle\}_{\alpha \in \Ac}$ satisfying $\langle \alpha | \alpha'\rangle = \delta(\alpha, \alpha')$. By contrast, a basis $\{|i\rangle\}_i \in \mH$ for separable $\mH$ will be indexed with natural numbers.  Frequently we will be working with states $|\Psi\rangle \in \mH$ that are completely specified by the sets of coefficients $\Psi(x) := \langle x| \Psi\rangle$ for elements $x$ of a (Borel) subset $\mathcal{X}\subset \R^d$, in which case we write $\Psi \in L^2(\mathcal{X})$ to denote square-integrable functions with respect to the measure $\I\{x \in \X\}$. A state $\rho \in \D(\mH)$ with $\mH = L^2(\X)$ in a continuous variable quantum system is naturally associated with a (positive) integral kernel operator $\rho: L^2(\X)\rightarrow L^2(\X)$ of the form $(\rho \Psi)(x) = \int_{x\in \X}dx \rho(x, x') \Psi(x')$ with $\langle x| \rho |x'\rangle = \rho(x, x')$ defined for all $x,x' \in \X$ \cite{integralkernel}. 

We will also be working with channels acting on separable Hilbert spaces. For now, we simply define the set of channels $\cptp(\mH, \mH')$ to be bounded linear maps $\mN: \mB(\mH)\rightarrow \mB(\mH')$ that act on a state $\Sigma \in \D(\mH)$ according to
\begin{equation}
    \mN(\Sigma) = \tr_{\mathcal{E}}(U(\Sigma \otimes \Psi)U^\dagger),
\end{equation}
where $\mH', \mathcal{E}$ are separable Hilbert spaces, $\Psi \in \D(\mathcal{E})$ is a pure state, and $U \in \U{\mH \otimes \mathcal{E}, \mH' \otimes \mathcal{E}}$ is an isometry. Further treatment of channels is delayed to Appendix~\ref{app:ent_recovery}.

\topic{Recall that the goal of a classical learner was to receive observations $B$ (which depended on $\alpha$) and then output an estimate $\hat{\alpha}$ that is classically correlated with $\alpha$. Analogously, in the entanglement fraction task a learner applies local operations to a system $B$ to maximize a type of quantum correlation between the resulting system $\hat{A}$ and some reference system $R$.} Consider separable Hilbert spaces $\mH_R, \mH_A,\mH_{\hat{A}}$, which may be taken to be $L^2(\Ac)$ without loss of generality, and some additional separable Hilbert space $\mH_B$. Fix a channel $\mathcal{N} \in \cptp(\mH_A, \mH_B)$, and consider all channels $\mathcal{D} \in \cptp( \mH_B, \mH_{\Ahat})$ acting on $B$. The quantum systems ($A$, $B$, and $\Ahat$) are each analogous to (continuous) random variables ($\alpha$, $B$, and $\hat{\alpha}$) in the classical learning setting: a quantum system $A$ evolves into a quantum system $B$ of ``observations'' via the map $\mN$, and then the learner prepares the state of system $\Ahat$ by applying $\mathcal{D}$ to system $B$. Unlike before, the initial state of the system is a potentially entangled state over the joint system $RA$. Given the system $RA$ prepared in some initial state $|\Psi\rangle_{RA}\in \mH_{RA}$,
\begin{align}
    |\Psi\rangle_{RA} &=  \int_{\Ac} dr  \int_{\Ac} d\alpha  \Psi_{RA}(r, \alpha) |r\rangle_R \otimes | \alpha \rangle_{A}\label{eq:psira} 
\end{align}
then system $B$ is prepared by applying $\mN$ to produce the state $\Sigma_{RB} = (\I \otimes \mN)(\dm{\Psi}_{RA})$. The goal of the learner is to find the $\mathcal{D}$ that maximizes the overlap of $(\I \otimes \mD)(\Sigma_{RB}) \in \D(\mH_{R\Ahat})$ with a particular (unnormalized) entangled state $|\Phi_\epsilon\rangle_{R\hat{A}} \in \mH_{R\hat{A}}$, defined as 
\begin{align} 
    |\Phi_\epsilon\rangle_{R\hat{A}} &= \int_{\Ac} dr \int_{\mathbb{B}_\epsilon(r)} d\hat{\alpha} |r\rangle_R \otimes  |\hat{\alpha}\rangle_{\hat{A}}\label{eq:phieps}
\end{align}
where $\mathbb{B}_\epsilon(r) := \{r' \in \Ac: d(r, r')\leq \epsilon\}$ denotes an $\epsilon$-ball centered at $r$ (restricted to $\Ac$). Thus, the learner's goal is to maximize the entanglement fraction of the final state on $R\Ahat$, which we might denote
\begin{equation}\label{eq:entanglement_fraction}
    Q_\epsilon(R|B)_\Sigma:= \langle \Phi_\epsilon| (\I \otimes  \mD)(\Sigma_{RB})|\Phi_\epsilon \rangle.
\end{equation}

\topic{Before proceeding, we justify the unusual choice of entangled state for this scenario.} A significant body of work has been devoted to characterizing entanglement in infinite dimensional systems \cite{eisert2002quantification,eisert2003introduction,kholevo2005notion,Adesso_2007}. In particular, many information-theoretic quantities in finite-dimensional systems have been at least partially extended to the infinite dimensional case \cite{furrer_min_2011,Regula_2021,Haapasalo_2021,yamasaki2024entanglement}, and such analyses could conceivably be extended to define a notion of maximal singlet fraction in infinite dimensions. As there is no well-defined notion of a maximally entangled state in a bipartite infinite-dimensional Hilbert space (barring the use of potentially unbounded operators \cite{Holevo_2011}), there is no obvious choice of entangled state with which to compute the fidelity of $(\I \otimes \mD )(\Sigma_{RB})$. However, we have chosen fidelity with the entangled state Eq.~\ref{eq:phieps} with the goal that this task to reduce to classical learning as a special case, via an argument involving dephasing that is analogous to the relationship between maximal singlet fraction and classical hypothesis testing provided in Sec.~\ref{sec:class_reduction}. In the next section we will demonstrate that the choice of entangled state $|\Phi_\epsilon\rangle$ satisfies this goal, and in this sense we will argue that above entanglement fraction task is a generalization of learning.

\subsection{The relationship to classical learning}\label{sec:class_reduction}

\topic{Just as optimal performance in classical hypothesis testing can be considered a special case of maximal singlet fraction in the presence of completely dephasing noise, we will now demonstrate how classical learning arises as a special case of the entanglement fraction maximization defined above.} The technique of Sec.~\ref{sec:cht_reduction} of applying a completely dephasing channel does not generalize to our purposes. The issue is that, while one could define a completely dephasing channel with respect to a countable basis $\{|i\rangle\} \subset \mH$ (e.g. Ref.~\cite{holevo2005separability}), this approach will not recover statements about probability distributions over parameters $\alpha \in \Ac$ that resemble classical learning of continuous parameters in any clear way. The basic problem is that any such basis is countable in $\mH$, but our learning task involves statements about probability densities $\mP(\Ac)$. There is no such basis for $\mH$ with components labelled by $\alpha \in \Ac$ that is not also overcomplete \cite{cholevo_statistical_2001}. Here, we will use an approach involving unbounded operators acting on $\mH$, though we will need to also employ several basic concepts from measure theory.

For $\mH = L^2(\Ac)$, denote the \textit{expectation value} $\langle X \rangle_\rho$ of a self-adjoint linear operator $X \in \mL(\mH)$ with respect to a state $\rho\in \D(\mH)$. Computed with respect to the set of linear functionals $\{\ket{\alpha}\}_{\alpha \in \Ac}$ that form a resolution of the identity, we have that $\langle X \rangle_\rho = \int_{\Ac} \int_{\Ac} d\alpha d\alpha' \langle \alpha | X | \alpha' \rangle \rho(\alpha, \alpha') \in \mathbb{R}$. Then, we define the (unbounded) map $\Delta_{\Ac}: \D(\mH)\rightarrow \mathcal{L}(\mH)$, the image of which has an expectation value that satisfies
\begin{equation}\label{eq:delta_A}
    \langle \Delta_{\Ac}(\rho)\rangle_{\sigma} = \int_{\Ac} d\alpha   \rho (\alpha, \alpha) \sigma(\alpha, \alpha).
\end{equation}
for $\rho, \sigma \in \D(\mH)$. If we associate the input state $\rho$ with a probability density $p \in L^1(\Ac)$ according to $p(\alpha):= \rho(\alpha, \alpha) = \langle \alpha | \rho |\alpha\rangle$ and associate $\sigma$ with $q \in L^1(\Ac)$ similarly, then Eq.~\ref{eq:delta_A} corresponds to the probability density for measurements of $\rho$ and $\sigma$ to agree. While the coefficients $\langle \alpha | \rho |\alpha\rangle$ corresponding to a set of operators $\{ \dm{\alpha}\}_{\alpha \in \Ac}$ do not correspond to outcomes of a valid measurement ($\{ \dm{\alpha}\}_{\alpha \in \Ac}$ is not a POVM), a valid POVM can be constructed by integrating the operators $\dm{\alpha}$ with respect to a partition of Borel subsets of $\Ac$. In this way, the quantities we will derive may be interpreted in terms of a measurement with respect to $\alpha \in \Ac$ while not themselves corresponding to any such measurement. Defining a probability density $p_R(r) = \int_\Ac d\alpha |\Psi(r, \alpha)|^2$ and writing the state
\begin{align}
    |\Psi^r\rangle_A = \frac{1}{\sqrt{p_R(r)}} \int_\Ac d\alpha \Psi(r, \alpha) |\alpha\rangle \in L^2(\Ac),
\end{align}
we find
\begin{align}
    \langle \Delta_{\Ac}(\Phi_\epsilon)\rangle_{\I \otimes \mathcal{D}\circ \mathcal{N}(\Psi_{RA})} &=  \int_\Ac dr p_R(r) \int_{\mathbb{B}_\epsilon(r)} d\hat{\alpha} \langle \hat{\alpha}| \mathcal{D}\circ \mathcal{N}(\Psi^r) |\hat{\alpha} \rangle
    \\&=  \int_\Ac dr p_R(r) \Pr_{p_{\Ahat|R}}( d(\hat{\alpha}, r) \leq \epsilon), \label{eq:diagonal_prob}
\end{align}
where we have identified a probability density $p_{\Ahat|R}(\hat{\alpha}|r) = \langle \hat{\alpha}| \mathcal{D}\circ \mathcal{N}(\Psi^r) |\hat{\alpha} \rangle$ for an outcome $\hat{\alpha}$ with respect to the collection of operators $\{ \dm{\hat{\alpha}}\}_{\alpha \in \Ac}$ applied to a state $\Sigma_{\Ahat}^r := \mathcal{D}\circ \mathcal{N}(\Psi^r) \in \D(\mH_{\Ahat})$. Substituting $r$ for $\alpha$, we recover the statement of Eq.~\ref{eq:class_wit} relating classical learning to the preparation of a variable $\hat{\alpha}$ (given observations in system $B$) that is highly correlated with the target parameter $\alpha$ via optimization over channels $\mathcal{D} \in \cptp(\mH_B, \mH_{\Ahat})$. Indeed, the operator $ \Delta_{\Ac}(\Phi_\epsilon)$ is essentially the indicator function $\I\{\hat{\alpha} \in \mathbb{B}_\epsilon(r)\}$ of Eq.~\ref{eq:class_wit} that is nonzero whenever the learner's estimate $\hat{\alpha}$ is $\epsilon$-close to a classical reference variable $r$ (which may be chosen arbitrarily close to $\alpha$ in the classical case).

The goal of the entanglement fraction task is for the learner to maximize $Q_\epsilon(R|B)_\Sigma$ with respect to channels $\mD \in \cptp(\mH_B, \mH_{\Ahat})$. We may understand the best-case performance that a learner might achieve at this task by considering the maximum of $Q_\epsilon(R|B)_\Sigma$ with respect to states $|\Psi_{RA}\rangle \in \mH_{RA}$. We will therefore compare maximization with respect to $|\Psi\rangle_{RA} \in \mH_{RA}$ and estimator channels to the maximum over \textit{parameters} $\alpha \in \Ac$ and stochastic maps in the classical case (e.g. Proposition~\ref{prop:maximax}). Since Eq.~\ref{eq:diagonal_prob} no longer explicitly depends on the system $A$, optimization over initial states and  channels $\mD$ recovers
\begin{align}
\sup_{|\Psi\rangle_{RA} \in \mH_{RA}} \sup_{\mathcal{D} \in \cptp(\mH_B, \mH_{\Ahat})} &\log \int_\Ac dr p_R(r) \Pr_{p_{\Ahat|R}}( d(\hat{\alpha}, r) \leq \epsilon)
\\&=\sup_{p_{R}\in \mP(\Ac)} \sup_{\hat{\alpha}: B \rightarrow \Ahat} \log \int_\Ac dr p_R(r) \Pr_{p_{B|R}}( d(\hat{\alpha}(B), r) \leq \epsilon)
\\&:= \esssup_{\alpha \in \Ac}\sup_{\hat{\alpha}: B \rightarrow \Ahat} \log \Pr_{p_{B|\alpha}}( d(\hat{\alpha}(B), \alpha) \leq \epsilon)
\end{align}
where the $\esssup$ indicates that the supremum is to be taken with respect to subsets $\{\alpha\}\subseteq \Ac$ having nonzero measure with respect to any probability measure defined on $\Ac$. This constraint is arguably more reasonable than a taking a supremum over all $\alpha \in \Ac$ by avoiding pathological scenarios, for instance, distributions of the unknown target parameter $\alpha$ that are everywhere discontinuous. Note that the system $B$ implicitly represents a quantum system, with $\hat{\alpha}$ being some composition of quantum channel followed by measurement - we discuss this limitation further in Sec.~\ref{sec:discussion}. 

\topic{We have shown that maximizing the entanglement fraction (the overlap with the particular entangled state $|\Phi_\epsilon\rangle$) reduces to maximizing classical correlations (the expected value of $\Delta_{\Ac}(\Phi_\epsilon)$) that correspond to classical learning}. Specifically, the choice of $\epsilon$ in $|\Phi_\epsilon\rangle$ corresponds to the error probability $\Pr(d(\alpha, \hat{\alpha}) \geq \epsilon)$ in some continuous parameter learning task, and $\dm{\Phi_\epsilon}$ is essentially a witness for quantum correlations that reduces the operator $\I\{\hat{\alpha} \in \mathbb{B}_\epsilon(\alpha)\}$ in Eq.~\ref{eq:class_wit} used to quantify classical correlations in the classical setting.

\subsection{Bounds for the entanglement fraction task}

\topic{We are interested in upper and lower bounds on the error of the entanglement fraction task, but we must first impose an additional constraint on the problem.} The state $\ket{\Phi_\epsilon}$ is both theoretically challenging to work with due to its possibly infinite Schmidt rank, and practically challenging since preparation of such a state in a lab may be infeasible. Instead, we  restrict our analysis to a more feasible scenario in which the entanglement in $\ket{\Phi_\epsilon}$ has been bounded by applying some bipartite projector $\Pi \in \mL(\mH_{R \Ahat})$. This choice also allows us to interpret the entanglement fraction in terms of teleportation directly.

\topic{We now state an upper and lower bound on the entanglement fraction task, via reduction to the maximal singlet fraction in finite dimensions}. Essentially the claim is that, a learner who acts locally via an \textit{optimal} channel $\mathcal{D}$ on system $B$ of a state $\sigma_{RB}$ to maximize overlap with $\Pi|\Phi_\epsilon\rangle$ will, in the \textit{best case} with respect to the initial state of $RA$, have a performance guarantee in terms of the conditional entropy of an appropriately chosen finite-dimensional system:

\begin{theorem}[Entanglement fraction guarantee, informal]\label{thm:ent_recovery}
    Consider a bounded subset $\Ac$ of a metric space $(X, d)$. For any $\epsilon > 0$, let $\mW:=\mW(\epsilon, \Ac, d)$ be an $\epsilon$-net of $\Ac$ in $d$ and let $C(\Ac, \mW)$ be the corresponding $\epsilon$-covering partition. Let $|w\rangle$ be constant on $R_w\in C(\Ac, \mW)$ and zero elsewhere in $\Ac$, and define 
    \begin{align}\label{eq:partition_singlet}
        |\phi\rangle_{RA} &= \frac{1}{|\mW|^{1/2}} \sum_{w=1}^{|\mW|} |w\rangle \otimes | w\rangle
    \end{align}
    For a subset $K := \{(r, \alpha) \in \Ac \times \Ac : \hat{W}(r) = \hat{W}(\alpha)\}$, define a projector $\Pi \in \mB(\mH_{R\Ahat})$ onto the subspace of $\mH_{R\Ahat}$ spanned by linear forms $|r\rangle \otimes |\alpha\rangle$ with $(r, \alpha) \in K$. Then, we have
    \begin{equation}
         \sup_{|\Psi\rangle_{RA} \in \mH_{RA}} \sup_{\mathcal{D} \in \cptp(\mH_B, \mH_{\Ahat})} \log \langle \Phi_\epsilon | \Pi(\I  \otimes \mathcal{D}) (\Sigma) \Pi| \Phi_\epsilon\rangle_{R\hat{A}} \geq  - \HHH(R|B)_\sigma 
    \end{equation}
    where $\Sigma = (\I \otimes \mN)(\dm{\Psi}_{RA}) \in \D(\mH_{RB})$ and $\HHH(R|B)_\sigma$ is computed with respect to $\sigma_{RB} = (\I \otimes \mathcal{N})(\dm{\phi}_{RA})$.

\end{theorem}
The proof is provided in Appendix~\ref{app:ent_recovery}. The basic idea may be understood as a communication protocol between Alice (system $RA$) and Bob (system $RB$) through a noisy channel $\mN$. While the parties may optimize over infinite-dimensional states and decoding channels respectively, they instead opt to use only finite dimensions - this results in a lower bound for their performance. The remaining details are in the specific sets of finite dimensional states (or channels, respectively) that Alice and Bob elect to use, which are informed by the $\epsilon$-covering partition chosen to establish the information-theoretic bound of Proposition~\ref{prop:maximax}. The key details are as follows: states that are piecewise-constant on $C(\Ac, \mW)$ are a subset of states in $\mH_{RA}$, the set of channels $\mD$ acting on $\mL(\mH_B)$ may be restricted to a subset acting nontrivially only on specific finite-dimensional subspaces, and the linearity of any fixed $\mN \in \cptp(\mH_A, \mH_B)$ restricts the dimension of the image of some subset of $\D(\mH_B)$ in a predictable way. Each of these properties is applied to bound either the infinum or supremum appearing in Theorem~\ref{thm:ent_recovery}, until the problem is completely described by states and channels acting on embedded finite-dimensional subspaces. The technical details are in setting up the discretizations of the relevant Hilbert spaces and picking specific operators and states to achieve the desired bound. In Sec.~\ref{sec:class_reduction} we have seen how this guarantee relates to classical learning, in that Theorem~\ref{thm:ent_recovery} resembles the classical learning guarantee of Proposition~\ref{prop:maximax} except for the presence of a score function $\score$.

\topic{Theorem~\ref{thm:ent_recovery} demonstrates how we can upper bound the error of the entanglement fraction task for a best-case scenario in terms of the conditional entropy of a particular finite-dimensional quantum state. We now discuss how we may lower bound the error of this task for a worst-case scenario in a similar way.} A key difference from the previous analysis is that the worst case entanglement fraction with respect to states $|\Psi\rangle_{RA} \in \mH_{RA}$ may be made arbitrarily small by preparing a \textit{separable} state of the form $|\Psi\rangle_{RA} = |\Psi_R\rangle \otimes |\Psi_A\rangle$. Since the learner's local operation $\mD$ cannot generate entanglement between $R$ and $\Ahat$, the resulting entanglement fraction $Q_\epsilon(R|B)_\Sigma$ may approach zero. Therefore we may only derive a non-trivial error lower bound if the space of input states is somehow constrained.

\topic{We will therefore only consider input states that have sufficiently large overlap with a maximally entangled state in some finite-dimensional subspace of $\mH_{RA}$}. This guarantees that, neglecting the effect of $\mN$, there is some fixed amount of entanglement initially present in the system. Briefly, for a state $|\Psi\rangle_{RA}$ we define $q_k(\Psi_{RA})$ to be the overlap of $|\Psi\rangle$ with $\frac{1}{\sqrt{k}}\sum_{i=1}^k |i_R\rangle \otimes |i_A\rangle$ for some choice of bases $\{|i_R\rangle\}_i \subset \mH_R$ and $\{|i_A\rangle\}_i\subset \mH_A$. To choose an appropriate $k$ for the task, we will once again consider a kind of discretization on our Hilbert space. As before, let $\mathcal{V} =\mathcal{V}(\epsilon, \Ac, d)$ be an $\epsilon$-packing of $\Ac$ in $d$. Letting $|v\rangle$ be constant within $\mathbb{B}_\epsilon(\alpha_v)$ and zero elsewhere, we can define a singlet-like state on $\mH_{RA}$ as
\begin{align}\label{eq:packing_singlet}
    |\phi\rangle_{RA} &= \frac{1}{|\mathcal{V}|^{1/2}} \sum_{v=1}^{|\mathcal{V}|} |v\rangle \otimes | v\rangle
\end{align}
Given that this is a maximally entangled state in $\C^{|\mathcal{V}|\times |\mathcal{V}|}$, a natural choice for the maximal singlet fraction of the initial state $|\Psi\rangle_{RA}$ is to choose a subspace of dimension $k = |\mathcal{V}|$. As before, we consider overlap with a state $\Pi|\Phi_\epsilon\rangle$ of bounded entanglement. The following then provides a upper bound on the entanglement fraction:

\begin{theorem}[Entanglement fraction bound, informal]\label{thm:ent_bound}
Let $|\Phi_\epsilon\rangle$, $\mathcal{D}$, and $\mathcal{N}$ be as defined in Theorem~\ref{thm:ent_recovery}, and define $q_k(\Psi)$ to be the singlet fraction of a state in $\mH_{RA}$ computed in any $(k\times k)$-dimensional subspaces of $\mH_{RA}$. Define a subset $K := \{(r, \alpha) \in \Ac \times \Ac : \hat{V}(r) = \hat{V}(\alpha)\} \subset \Ac \times \Ac$, where $\hat{V}$ returns the nearest point to $r$ in the $\epsilon$-packing $\mV$, and let $\Pi \in \mB(\mH_{R\Ahat})$ be a projector onto the subspace of $\mH_{R\Ahat}$ spanned by unbounded linear functionals $|r\rangle \otimes |\alpha\rangle$ with $(r, \alpha) \in K$. Then,
\begin{equation}\label{eq:ent_bound}
      \inf_{\substack{|\Psi\rangle_{RA}\in \mH_{RA} \\ q_{|\mathcal{V}|}(\Psi) = 1}} \sup_{\mathcal{D} \in \cptp(\mH_B, \mH_{\Ahat})} \bigl( 1 - |\mV|^{-1}\langle \Phi_\epsilon | \Pi  (\I\otimes \mathcal{D})(\Sigma) \Pi |\Phi_\epsilon \rangle \bigr) \geq \frac{\HHH(RB)_{\rho} - 1}{\log(|\mathcal{V}|^2 - 1)} 
\end{equation}
where  $\Sigma = (\I \otimes \mathcal{N} )(\dm{\Psi})$ and $\rho = (\I \otimes \mathcal{N})(\dm{\phi})$ for $|\phi\rangle$ defined in Eq.~\ref{eq:packing_singlet}.

\end{theorem}
The proof is provided in Appendix~\ref{app:ent_bound}. Similarly to Theorem~\ref{thm:ent_recovery}, the LHS of Eq.~\ref{eq:ent_bound} is analogous to $\Pr(d(\alpha, \hat{\alpha}) \geq \epsilon)$ in Proposition~\ref{prop:minimax}, without any reference to the loss function $\ell$. The bounded expression (restricted to $[0,1]$) can be interpreted as a normalized error, which becomes larger as the maximal entanglement fraction decreases. Applying an overall factor of $|\mV|$ (corresponding to the chosen packing $\mV$), Theorem~\ref{thm:ent_bound} provides a lower bound for the error of the entanglement fraction task.

\section{Discussion}\label{sec:discussion}

\topic{In this work, we have demonstrated how classical learning can be understood as a special case of a particular entanglement manipulation task in a continuous variable quantum system.} The task that we have introduced is distinct from the typical setting of quantum machine learning, in that it involves coherent manipulation of quantum information at all stages. We derived upper and lower bounds for the error of this task via reduction to maximal singlet fraction in finite dimensions, in an analogous fashion to the reduction of classical learning to multi-hypothesis testing. 

\topic{A central challenge remains in \textit{applying} learning guarantees (such as Proposition~\ref{prop:maximax}) and analogous entanglement fraction guarantee (Theorem~\ref{thm:ent_recovery})}. Applying the former relies on upper bounding the conditional entropy $\HHH(W|B)$ of a variable indexed by the $\epsilon$-net for a particular distribution of observations $p_{B|A}$. For context, when a learner's observations are independent samples $B=X_1, \dots, X_n$, then minimax bounds of the form of Proposition~\ref{prop:minimax} may be combined with the convexity of relative entropy to lower bound $\HHH(V|B)$ in terms of marginal entropies \cite{loh_corrupted_2012}. Applying the entanglement fraction guarantee may require identifying a problem setup where similar simplifications hold.

\topic{There are key differences between the entanglement fraction bounds and their classical analogues, with significant technical challenges appearing in the former case.} For instance, there is no obvious analogue for the loss function $\ell$ and score function $\score$ applied to the distance $d(\hat{\alpha}, \alpha)$ in the classical setting. One might introduce some unbounded operator $\dm{L} \in \mL(\mH_{R \hat{A}})$ to fulfill an analogous role to the loss function $\ell$, for instance defining $|L\rangle$ such that $\langle \alpha, \hat{\alpha} | L\rangle =  \ell^{1/2}(d(\alpha,  \hat{\alpha}))$ for $\alpha, \hat{\alpha} \in \Ac$. However,  the non-commutativity of quantum theory obscures how the expected value of such an observable is related to a state's overlap with $\ket{\Phi_\epsilon}$. Future work may explore alternative generalizations of classical learning that can better accommodate loss functions. In Sec.~\ref{sec:class_reduction} we discussed how the entanglement fraction task reduced to a classical learning task, albeit one involving quantum input data. A reduction to the purely classical case might involve technical hurdles. While applying the completely dephasing map $\Delta$ to the system $RB$ in the finite-dimensional case (Eq.~\ref{eq:finite_dephase}) is straightforward, a similar procedure could not be followed for the infinite-dimensional case due to the unboundedness of the map $\Delta_{\Ac}$. Equivalently, evaluating $\langle \Delta_{\Ac}\circ(\I \otimes \mD^*)(\Phi_\epsilon)\rangle_{(\I\otimes \mN)\Psi_{RA}}$ involves products of unbounded operators, which raises difficulties. In any case, we have imposed a reduction to classical learning by exploiting the commutativity of measurements with respect to a spectral measure constructed from linear functionals $\{|\alpha\rangle\}_{\alpha \in \Ac}$, at the expense of constraining measurement statistics in a conjugate basis. A more complete theory should ideally involve genuinely non-classical measurement statistics, i.e. constraints on the measurements of incompatible observables. 

\topic{In this work we have focused on a specific type of learning task involving parameter estimation, which differs from the task of predicting the value of a target function applied to a set of observations -- another typical task in machine learning.} Our analysis generalizes to the case when the observations are produced by a function $f_\alpha$ parameterized by $\alpha \in \Ac$. However, in more general learning models the data may be produced by a process for which there is no explicit target function $f_\alpha$ (as is the case in agnostic PAC learning \cite{haussler_decision_1992}). It remains an open question how we might view more general prediction tasks through a similar framework as the one developed here.

\section{Acknowledgements}

Thank you to Achim Kempf, T. Rick Perche, Jose Polo Gomez, Arsalan Motamedi, Kohdai Kuroiwa, Ernest Tan, and Jon Yard for helpful discussions and reviewing drafts of this manuscript. Circuit diagrams were generated using the \textsc{quantikz} circuit library \cite{quantikz}. Research at Perimeter Institute is supported in part by the Government of Canada through the Department of Innovation, Science and Economic Development and by the Province of Ontario through the Ministry of Colleges and Universities.

\clearpage
\printbibliography

@article{feder_relations_1994,
	title = {Relations between entropy and error probability},
	volume = {40},
	doi = {10.1109/18.272494},
	pages = {259--266},
	number = {1},
	journaltitle = {{IEEE} Trans. Inf. Theory},
	author = {Feder, M. and Merhav, N.},
	date = {1994-01},
}

@book{attal2014quantum,
  title={Lectures in Quantum Noise Theory},
  author={Attal, Stephane},
  url={http://math.univ-lyon1.fr/~attal/chapters.html},
}

@article{stinespring,
 ISSN = {00029939, 10886826},
 URL = {http://www.jstor.org/stable/2032342},
 author = {W. Forrest Stinespring},
 journal = {Proceedings of the American Mathematical Society},
 number = {2},
 pages = {211--216},
 publisher = {American Mathematical Society},
 title = {Positive Functions on C*-Algebras},
 urldate = {2024-03-03},
 volume = {6},
 year = {1955}
}

@BOOK{kraus,
       author = {{Kraus}, Karl and {B{\"o}hm}, A. and {Dollard}, J.~D. and {Wootters}, W.~H.},
        title = "{States, Effects, and Operations Fundamental Notions of Quantum Theory}",
         year = 1983,
       volume = {190},
          doi = {10.1007/3-540-12732-1},
       adsurl = {https://ui.adsabs.harvard.edu/abs/1983LNP...190.....K}
}

@misc{watrous2018advanced,
  author = {John Watrous},
  title = {Advanced topics in Quantum Information Theory},
  year = {2020},
  howpublished = { Lecture Notes},
  note = {Available at \url{https://cs.uwaterloo.ca/~watrous/QIT-notes/}},
}

@article{koenig_operational_2009,
	title = {The operational meaning of min- and max-entropy},
	volume = {55},
	url = {http://arxiv.org/abs/0807.1338},
	doi = {10.1109/TIT.2009.2025545},
	abstract = {We show that the conditional min-entropy Hmin(A{\textbar}B) of a bipartite state rho\_AB is directly related to the maximum achievable overlap with a maximally entangled state if only local actions on the B-part of rho\_AB are allowed. In the special case where A is classical, this overlap corresponds to the probability of guessing A given B. In a similar vein, we connect the conditional max-entropy Hmax(A{\textbar}B) to the maximum fidelity of rho\_AB with a product state that is completely mixed on A. In the case where A is classical, this corresponds to the security of A when used as a secret key in the presence of an adversary holding B. Because min- and max-entropies are known to characterize information-processing tasks such as randomness extraction and state merging, our results establish a direct connection between these tasks and basic operational problems. For example, they imply that the (logarithm of the) probability of guessing A given B is a lower bound on the number of uniform secret bits that can be extracted from A relative to an adversary holding B.},
	pages = {4337--4347},
	number = {9},
	journaltitle = {{IEEE} Trans. Inf. Theory},
	shortjournal = {{IEEE} Trans. Inform. Theory},
	author = {Koenig, Robert and Renner, Renato and Schaffner, Christian},
	urldate = {2023-02-27},
	date = {2009-09},
	% eprinttype = {arxiv},
	% eprint = {0807.1338 [quant-ph]},
	keywords = {Quantum Physics},
}

@article{kovalevsky1968problem,
  title={The problem of character recognition from the point of view of mathematical statistics},
  journal={Character Readers and Pattern Recognition},
  pages={3--30},
  year={1968},
  publisher={Spartan New York},
  author={Vladimir Kovalevsky},
}

@article{hellman_probability_1970,
	title = {Probability of error, equivocation, and the Chernoff bound},
	volume = {16},
	issn = {1557-9654},
	doi = {10.1109/TIT.1970.1054466},
	pages = {368--372},
	number = {4},
	journaltitle = {{IEEE} Trans. Inf. Theory},
	author = {Hellman, M. and Raviv, J.},
	date = {1970},
 }

@book{cover1999elements,
  title={Elements of information theory},
  author={Cover, Thomas M},
  year={1999},
  publisher={John Wiley \& Sons}
}

@book{cholevo_statistical_2001,
	location = {Berlin Heidelberg},
	title = {Statistical structure of quantum theory},
	isbn = {978-3-540-42082-8},
	series = {Lecture Notes in Physics Monographs},
	number = {67},
	publisher = {Springer},
	author = {Holevo, A. S.},
	date = {2001},
	langid = {english},
	doi = {https://doi.org/10.1007/3-540-44998-1},
}

@article{furrer_min_2011,
	title = {Min- and Max-Entropy in Infinite Dimensions},
	volume = {306},
	url = {https://doi.org/10.1007/s00220-011-1282-1},
	doi = {10.1007/s00220-011-1282-1},
	pages = {165--186},
	number = {1},
	journaltitle = {Comm. Math. Phys.},
	shortjournal = {Commun. Math. Phys.},
	author = {Furrer, Fabian and Åberg, Johan and Renner, Renato},
	urldate = {2023-08-10},
	date = {2011-08-01},
	langid = {english},
}

@book{watrous2018theory,
  title={The theory of quantum information},
  author={Watrous, John},
  year={2018},
  publisher={Cambridge university press}
}

@misc{quantikz,
  doi = {10.17637/RH.7000520},
  url = {https://royalholloway.figshare.com/articles/Quantikz/7000520},
  author = {Kay, Alastair},
  title = {Quantikz},
  publisher = {Royal Holloway, University of London},
  year = {2019},
  copyright = {Creative Commons Attribution 4.0 International}
}

@article{choi_completely_1975,
	title = {Completely positive linear maps on complex matrices},
	volume = {10},
	issn = {0024-3795},
	url = {https://www.sciencedirect.com/science/article/pii/0024379575900750},
	doi = {https://doi.org/10.1016/0024-3795(75)90075-0},
	pages = {285--290},
	number = {3},
	journaltitle = {Linear Algebra and its Applications},
	author = {Choi, Man-Duen},
	date = {1975},
}

@article{jamiolkowski_linear_1972,
	title = {Linear transformations which preserve trace and positive semidefiniteness of operators},
	volume = {3},
	issn = {0034-4877},
	url = {https://www.sciencedirect.com/science/article/pii/0034487772900110},
	doi = {https://doi.org/10.1016/0034-4877(72)90011-0},
	pages = {275--278},
	number = {4},
	journaltitle = {Reports on Mathematical Physics},
	author = {Jamiołkowski, A.},
	date = {1972},
}

@MISC {4920,
    TITLE = {Is there any connection between the diamond norm and the distance of the associated states?},
    AUTHOR = {John Watrous },
    HOWPUBLISHED = {Theoretical Computer Science Stack Exchange},
    URL = {},
    note = {\url{https://cstheory.stackexchange.com/q/4920}}
    }

@article{raskutti_minimax_2011,
	title = {Minimax Rates of Estimation for High-Dimensional Linear Regression Over $\ell_q$ -Balls},
	volume = {57},
	doi = {10.1109/TIT.2011.2165799},
	pages = {6976--6994},
	number = {10},
	journaltitle = {{IEEE} Trans. Inf. Theory},
	author = {Raskutti, Garvesh and Wainwright, Martin J. and Yu, Bin},
	date = {2011-10},
}

@article{caro2022outofdistribution,
   title={Out-of-distribution generalization for learning quantum dynamics},
   volume={14},
   ISSN={2041-1723},
   url={http://dx.doi.org/10.1038/s41467-023-39381-w},
   DOI={10.1038/s41467-023-39381-w},
   number={1},
   journal={Nat. Commun.},
   publisher={Springer Science and Business Media LLC},
   author={Caro, Matthias C. and Huang, Hsin-Yuan and Ezzell, Nicholas and Gibbs, Joe and Sornborger, Andrew T. and Cincio, Lukasz and Coles, Patrick J. and Holmes, Zoë},
   year={2023},
   month=jul }

@misc{fano_inequality,
  author        = {Robert Fano},
  title         = {Class notes for MIT course 6.574: Transmission of information},
  month         = {2},
  year          = {1952},
  publisher={MIT}
}

@article{quantumfano,
  title = {Sending entanglement through noisy quantum channels},
  author = {Schumacher, Benjamin},
  journal = {Phys. Rev. A},
  volume = {54},
  issue = {4},
  pages = {2614--2628},
  numpages = {0},
  year = {1996},
  month = {10},
  publisher = {American Physical Society},
  doi = {10.1103/PhysRevA.54.2614},
  url = {https://link.aps.org/doi/10.1103/PhysRevA.54.2614}
}

@ARTICLE{shanon1948,
  author={Shannon, C. E.},
  journal={The Bell System Technical Journal}, 
  title={A mathematical theory of communication}, 
  year={1948},
  volume={27},
  number={3},
  pages={379-423},
  doi={10.1002/j.1538-7305.1948.tb01338.x}}

@article{huangprl,
  title = {Information-Theoretic Bounds on Quantum Advantage in Machine Learning},
  author = {Huang, Hsin-Yuan and Kueng, Richard and Preskill, John},
  journal = {Phys. Rev. Lett.},
  volume = {126},
  issue = {19},
  pages = {190505},
  numpages = {7},
  year = {2021},
  month = {5},
  publisher = {American Physical Society},
  doi = {10.1103/PhysRevLett.126.190505},
  url = {https://link.aps.org/doi/10.1103/PhysRevLett.126.190505}
}

@article{huang2021power,
  title={Power of data in quantum machine learning},
  author={Huang, Hsin-Yuan and Broughton, Michael and Mohseni, Masoud and Babbush, Ryan and Boixo, Sergio and Neven, Hartmut and McClean, Jarrod R},
  journal={Nat. Commun.},
  volume={12},
  number={1},
  pages={2631},
  year={2021},
  publisher={Nature Publishing Group UK London}
}

@article{kubler2021inductive,
      title={The Inductive Bias of Quantum Kernels}, 
      author={Jonas M. Kübler and Simon Buchholz and Bernhard Schölkopf},
      year={2021},
      journal={Advances in Neural Information Processing Systems 34 (NeurIPS 2021)},
      pages={12661-12673},
      eprint={2106.03747},
      archivePrefix={arXiv},
      primaryClass={quant-ph}
}

@misc{caro2023learning,
      title={Learning Quantum Processes and Hamiltonians via the Pauli Transfer Matrix}, 
      author={Matthias C. Caro},
      year={2023},
      eprint={2212.04471},
      archivePrefix={arXiv},
      primaryClass={quant-ph}
}

@article{Khatri2019quantumassisted,
  doi = {10.22331/q-2019-05-13-140},
  url = {https://doi.org/10.22331/q-2019-05-13-140},
  title = {Quantum-assisted quantum compiling},
  author = {Khatri, Sumeet and LaRose, Ryan and Poremba, Alexander and Cincio, Lukasz and Sornborger, Andrew T. and Coles, Patrick J.},
  journal = {{Quantum}},
  issn = {2521-327X},
  publisher = {{Verein zur F{\"{o}}rderung des Open Access Publizierens in den Quantenwissenschaften}},
  volume = {3},
  pages = {140},
  month = may,
  year = {2019}
}

@article{Caro_2023v,
   title={Out-of-distribution generalization for learning quantum dynamics},
   volume={14},
   ISSN={2041-1723},
   url={http://dx.doi.org/10.1038/s41467-023-39381-w},
   DOI={10.1038/s41467-023-39381-w},
   number={1},
   journal={Nat. Commun.},
   publisher={Springer Science and Business Media LLC},
   author={Caro, Matthias C. and Huang, Hsin-Yuan and Ezzell, Nicholas and Gibbs, Joe and Sornborger, Andrew T. and Cincio, Lukasz and Coles, Patrick J. and Holmes, Zoë},
   year={2023},
   month= 7
}

@article{Gyurik2022towardsquantum,
  doi = {10.22331/q-2022-11-10-855},
  url = {https://doi.org/10.22331/q-2022-11-10-855},
  title = {Towards quantum advantage via topological data analysis},
  author = {Gyurik, Casper and Cade, Chris and Dunjko, Vedran},
  journal = {{Quantum}},
  issn = {2521-327X},
  publisher = {{Verein zur F{\"{o}}rderung des Open Access Publizierens in den Quantenwissenschaften}},
  volume = {6},
  pages = {855},
  month = 11,
  year = {2022}
}

@article{peters2022generalization,
   title={Generalization despite overfitting in quantum machine learning models},
   volume={7},
   ISSN={2521-327X},
   url={http://dx.doi.org/10.22331/q-2023-12-20-1210},
   DOI={10.22331/q-2023-12-20-1210},
   journal={Quantum},
   publisher={Verein zur Forderung des Open Access Publizierens in den Quantenwissenschaften},
   author={Peters, Evan and Schuld, Maria},
   year={2023},
   month=12, pages={1210}
}

@article{PRXQuantum.2.040321,
  title = {Generalization in Quantum Machine Learning: A Quantum Information Standpoint},
  author = {Banchi, Leonardo and Pereira, Jason and Pirandola, Stefano},
  journal = {PRX Quantum},
  volume = {2},
  issue = {4},
  pages = {040321},
  numpages = {21},
  year = {2021},
  month = {11},
  publisher = {American Physical Society},
  doi = {10.1103/PRXQuantum.2.040321},
  url = {https://link.aps.org/doi/10.1103/PRXQuantum.2.040321}
}

@article{schumacher_coding,
  title = {Quantum coding},
  author = {Schumacher, Benjamin},
  journal = {Phys. Rev. A},
  volume = {51},
  issue = {4},
  pages = {2738--2747},
  numpages = {0},
  year = {1995},
  month = {4},
  publisher = {American Physical Society},
  doi = {10.1103/PhysRevA.51.2738},
  url = {https://link.aps.org/doi/10.1103/PhysRevA.51.2738}
}

@misc{jerbi2023power,
      title={The power and limitations of learning quantum dynamics incoherently}, 
      author={Sofiene Jerbi and Joe Gibbs and Manuel S. Rudolph and Matthias C. Caro and Patrick J. Coles and Hsin-Yuan Huang and Zoë Holmes},
      year={2023},
      eprint={2303.12834},
      archivePrefix={arXiv},
      primaryClass={quant-ph}
}

@article{huang_arbitrary,
  title = {Learning to Predict Arbitrary Quantum Processes},
  author = {Huang, Hsin-Yuan and Chen, Sitan and Preskill, John},
  journal = {PRX Quantum},
  volume = {4},
  issue = {4},
  pages = {040337},
  numpages = {44},
  year = {2023},
  month = {12},
  publisher = {American Physical Society},
  doi = {10.1103/PRXQuantum.4.040337},
  url = {https://link.aps.org/doi/10.1103/PRXQuantum.4.040337}
}

@article{Huang_2022,
   title={Provably efficient machine learning for quantum many-body problems},
   volume={377},
   ISSN={1095-9203},
   url={http://dx.doi.org/10.1126/science.abk3333},
   DOI={10.1126/science.abk3333},
   number={6613},
   journal={Science},
   publisher={American Association for the Advancement of Science (AAAS)},
   author={Huang, Hsin-Yuan and Kueng, Richard and Torlai, Giacomo and Albert, Victor V. and Preskill, John},
   year={2022},
   month=sep }

@misc{gyurik2023exponential,
      title={Exponential separations between classical and quantum learners}, 
      author={Casper Gyurik and Vedran Dunjko},
      year={2023},
      eprint={2306.16028},
      archivePrefix={arXiv},
      primaryClass={quant-ph}
}

@article{Liu_2021,
   title={A rigorous and robust quantum speed-up in supervised machine learning},
   volume={17},
   ISSN={1745-2481},
   url={http://dx.doi.org/10.1038/s41567-021-01287-z},
   DOI={10.1038/s41567-021-01287-z},
   number={9},
   journal={Nature Physics},
   publisher={Springer Science and Business Media LLC},
   author={Liu, Yunchao and Arunachalam, Srinivasan and Temme, Kristan},
   year={2021},
   month=jul, pages={1013–1017} }

@article{gilfuster2024understanding,
   title={Understanding quantum machine learning also requires rethinking generalization},
   volume={15},
   ISSN={2041-1723},
   url={http://dx.doi.org/10.1038/s41467-024-45882-z},
   DOI={10.1038/s41467-024-45882-z},
   number={1},
   journal={Nat. Commun.},
   publisher={Springer Science and Business Media LLC},
   author={Gil-Fuster, Elies and Eisert, Jens and Bravo-Prieto, Carlos},
   year={2024},
   month=mar }

@article{Huang_2021,
   title={Power of data in quantum machine learning},
   volume={12},
   ISSN={2041-1723},
   url={http://dx.doi.org/10.1038/s41467-021-22539-9},
   DOI={10.1038/s41467-021-22539-9},
   number={1},
   journal={Nat. Commun.},
   publisher={Springer Science and Business Media LLC},
   author={Huang, Hsin-Yuan and Broughton, Michael and Mohseni, Masoud and Babbush, Ryan and Boixo, Sergio and Neven, Hartmut and McClean, Jarrod R.},
   year={2021},
   month=may }

@article{exact_learning,
  title={Queries and concept learning},
  author={Angluin, Dana},
  journal={Machine learning},
  volume={2},
  pages={319--342},
  year={1988},
  publisher={Springer}
}

@article{hhl_2009,
  title = {Quantum Algorithm for Linear Systems of Equations},
  author = {Harrow, Aram W. and Hassidim, Avinatan and Lloyd, Seth},
  journal = {Phys. Rev. Lett.},
  volume = {103},
  issue = {15},
  pages = {150502},
  numpages = {4},
  year = {2009},
  month = {10},
  publisher = {American Physical Society},
  doi = {10.1103/PhysRevLett.103.150502},
  url = {https://link.aps.org/doi/10.1103/PhysRevLett.103.150502}
}

@inproceedings{loh_corrupted_2012,
	location = {Cambridge, {MA}, {USA}},
	title = {Corrupted and missing predictors: Minimax bounds for high-dimensional linear regression},
	url = {http://ieeexplore.ieee.org/document/6283989/},
	doi = {10.1109/ISIT.2012.6283989},
	shorttitle = {Corrupted and missing predictors},
	eventtitle = {2012 {IEEE} International Symposium on Information Theory - {ISIT}},
	pages = {2601--2605},
	booktitle = {2012 {IEEE} International Symposium on Information Theory Proceedings},
	publisher = {{IEEE}},
	author = {Loh, Po-Ling and Wainwright, Martin J.},
	date = {2012-07},
}

@INPROCEEDINGS{wang2010information,
  author={Wei Wang and Wainwright, Martin J. and Ramchandran, Kannan},
  booktitle={2008 IEEE International Symposium on Information Theory}, 
  title={Information-theoretic limits on sparse support recovery: Dense versus sparse measurements}, 
  year={2008},
  volume={},
  number={},
  pages={2197-2201},
  doi={10.1109/ISIT.2008.4595380}}

@article{gross_2010,
  title = {Quantum State Tomography via Compressed Sensing},
  author = {Gross, David and Liu, Yi-Kai and Flammia, Steven T. and Becker, Stephen and Eisert, Jens},
  journal = {Phys. Rev. Lett.},
  volume = {105},
  issue = {15},
  pages = {150401},
  numpages = {4},
  year = {2010},
  month = {10},
  publisher = {American Physical Society},
  doi = {10.1103/PhysRevLett.105.150401},
  url = {https://link.aps.org/doi/10.1103/PhysRevLett.105.150401}
}

@ARTICLE{compressedsensing1,

  author={Candes, Emmanuel J. and Tao, Terence},

  journal={{IEEE} Trans. Inf. Theory}, 

  title={Near-Optimal Signal Recovery From Random Projections: Universal Encoding Strategies?}, 

  year={2006},

  volume={52},

  number={12},

  pages={5406-5425},

  doi={10.1109/TIT.2006.885507}}

@ARTICLE{compressedsensing2,

  author={Donoho, D.L.},

  journal={{IEEE} Trans. Inf. Theory}, 

  title={Compressed sensing}, 

  year={2006},

  volume={52},

  number={4},

  pages={1289-1306},

  doi={10.1109/TIT.2006.871582}}

@article{bennet_1996b,
  title = {Mixed-state entanglement and quantum error correction},
  author = {Bennett, Charles H. and DiVincenzo, David P. and Smolin, John A. and Wootters, William K.},
  journal = {Phys. Rev. A},
  volume = {54},
  issue = {5},
  pages = {3824--3851},
  numpages = {0},
  year = {1996},
  month = {11},
  publisher = {American Physical Society},
  doi = {10.1103/PhysRevA.54.3824},
  url = {https://link.aps.org/doi/10.1103/PhysRevA.54.3824}
}

@article{candes2012exact,
  title={Exact matrix completion via convex optimization},
  author={Candes, Emmanuel and Recht, Benjamin},
  journal={Found Comput Math},
  volume={ 9},
  pages={717–772},
  year={2009},
  doi={https://doi.org/10.1007/s10208-009-9045-5},
}

@article{eisert2002quantification,
  title={On the quantification of entanglement in infinite-dimensional quantum systems},
  author={Eisert, Jens and Simon, Christoph and Plenio, Martin B},
  journal={Journal of Physics A: Mathematical and General},
  volume={35},
  number={17},
  pages={3911},
  year={2002},
  publisher={IOP Publishing}
}

@article{eisert2003introduction,
  title={Introduction to the basics of entanglement theory in continuous-variable systems},
  author={J. Eisert and M. B. Plenio},
  journal={International Journal of Quantum Information},
  volume={1},
  number={04},
  pages={479--506},
  year={2003},
  publisher={World Scientific}
}

@article{kholevo2005notion,
  title={On the notion of entanglement in Hilbert spaces},
  author={Holevo, A. S. and Shirokov, Maksim and Werner, Reinhard},
  journal={Russian Mathematical Surveys},
  volume={60},
  number={2},
  pages={359--360},
  year={2005},
  publisher={London: London Mathematical Society; distributed by Cleaver-Hume Press,[1960-}
}

@article{Holevo_2011,
   title={The Choi–Jamiolkowski forms of quantum Gaussian channels},
   volume={52},
   ISSN={1089-7658},
   url={http://dx.doi.org/10.1063/1.3581879},
   DOI={10.1063/1.3581879},
   number={4},
   journal={Journal of Mathematical Physics},
   publisher={AIP Publishing},
   author={Holevo, A. S.},
   year={2011},
   month=apr }

@article{Adesso_2007,
doi = {10.1088/1751-8113/40/28/S01},
year = {2007},
month = {6},
publisher = {},
volume = {40},
number = {28},
pages = {7821},
author = {Gerardo Adesso and Fabrizio Illuminati},
title = {Entanglement in continuous-variable systems: recent advances and current perspectives},
journal = {Journal of Physics A: Mathematical and Theoretical},
}

@misc{yamasaki2024entanglement,
      title={Entanglement cost for infinite-dimensional physical systems}, 
      author={Hayata Yamasaki and Kohdai Kuroiwa and Patrick Hayden and Ludovico Lami},
      year={2024},
      eprint={2401.09554},
      archivePrefix={arXiv},
      primaryClass={quant-ph}
}

@ARTICLE{candes2009power,
  author={Candes, Emmanuel J. and Tao, Terence},
  journal={{IEEE} Trans. Inf. Theory}, 
  title={The Power of Convex Relaxation: Near-Optimal Matrix Completion}, 
  year={2010},
  volume={56},
  number={5},
  pages={2053-2080},
  doi={10.1109/TIT.2010.2044061}}

@article{bennet_1996a,
  title = {Purification of Noisy Entanglement and Faithful Teleportation via Noisy Channels},
  author = {Bennett, Charles H. and Brassard, Gilles and Popescu, Sandu and Schumacher, Benjamin and Smolin, John A. and Wootters, William K.},
  journal = {Phys. Rev. Lett.},
  volume = {76},
  issue = {5},
  pages = {722--725},
  numpages = {0},
  year = {1996},
  month = {1},
  publisher = {American Physical Society},
  doi = {10.1103/PhysRevLett.76.722},
  url = {https://link.aps.org/doi/10.1103/PhysRevLett.76.722}
}

@article{Verstraete_2002,
   title={Fidelity of mixed states of two qubits},
   volume={66},
   ISSN={1094-1622},
   url={http://dx.doi.org/10.1103/PhysRevA.66.022307},
   DOI={10.1103/physreva.66.022307},
   number={2},
   journal={Phys. Rev. A},
   publisher={American Physical Society (APS)},
   author={Verstraete, Frank and Verschelde, Henri},
   year={2002},
   month=aug }

@proceedings{QIP:24,
  editor =       {},
  title =        "27th Conference on Quantum Information Processing (QIP 2024) ",
  booktitle =    "27th Conference on Quantum Information Processing (QIP 2024)",
  publisher =    {},
  venue =        {},
  month =        1,
  year =         {2024},
  isbn =         {},
}

@inproceedings{meyer2023quantum:QIP:24,
    title={Quantum metrology in the finite-sample regime}, 
      author={Johannes Jakob Meyer and Sumeet Khatri and Daniel Stilck França and Jens Eisert and Philippe Faist},
    booktitle={Quantum Information Processing 2024},
    year = 2024,
      eprint={2307.06370},
      archivePrefix={arXiv},
      primaryClass={quant-ph},
    crossref =     {QIP:24},
    doi={},

}

@article{qpac2,
   title={Quantum Algorithms for Learning and Testing Juntas},
   volume={6},
   ISSN={1573-1332},
   url={http://dx.doi.org/10.1007/s11128-007-0061-6},
   DOI={10.1007/s11128-007-0061-6},
   number={5},
   journal={Quantum Information Processing},
   publisher={Springer Science and Business Media LLC},
   author={Atıcı, Alp and Servedio, Rocco A.},
   year={2007},
   month=sep, pages={323–348} }

@article{qpac1,
author = {Servedio, Rocco A. and Gortler, Steven J.},
title = {Equivalences and Separations Between Quantum and Classical Learnability},
journal = {SIAM Journal on Computing},
volume = {33},
number = {5},
pages = {1067-1092},
year = {2004},
doi = {10.1137/S0097539704412910},
}

@article{nfl_2022,
  title = {Reformulation of the No-Free-Lunch Theorem for Entangled Datasets},
  author = {Sharma, Kunal and Cerezo, M. and Holmes, Zo\"e and Cincio, Lukasz and Sornborger, Andrew and Coles, Patrick J.},
  journal = {Phys. Rev. Lett.},
  volume = {128},
  issue = {7},
  pages = {070501},
  numpages = {7},
  year = {2022},
  month = {2},
  publisher = {American Physical Society},
  doi = {10.1103/PhysRevLett.128.070501},
  url = {https://link.aps.org/doi/10.1103/PhysRevLett.128.070501}
}

@InProceedings{kothari2014,
  author =	{Robin Kothari},
  title =	{{An optimal quantum algorithm for the oracle identification problem}},
  booktitle =	{31st International Symposium on Theoretical Aspects of Computer Science (STACS 2014)},
  pages =	{482--493},
  ISBN =	{978-3-939897-65-1},
  volume =	{25},
  publisher =	{Schloss Dagstuhl--Leibniz-Zentrum fuer Informatik},
  doi =		{10.4230/LIPIcs.STACS.2014.482},
}

@inproceedings{ambainis_quantum_2004,
	title = {Quantum Identification of Boolean Oracles},
	pages = {105--116},
   booktitle = {Proceedings of the 21st International Symposium on Theoretical Aspects of Computer Science (STACS 2004)},
	author = {Ambainis, Andris and Iwama, Kazuo and Kawachi, Akinori and Masuda, Hiroyuki and Putra, Raymond H. and Yamashita, Shigeru},
    doi = {https://doi.org/10.1007/978-3-540-24749-4_10},
    publisher = {Springer, Berlin, Heidelberg},
}

@article{arunachalam2018optimal,
  title={Optimal quantum sample complexity of learning algorithms},
  author={Arunachalam, Srinivasan and De Wolf, Ronald},
  journal={The Journal of Machine Learning Research},
  volume={19},
  number={1},
  pages={2879--2878},
  year={2018},
  publisher={JMLR. org}
}

@inproceedings{bj99,
author = {Bshouty, Nader H. and Jackson, Jeffrey C.},
title = {Learning DNF over the Uniform Distribution Using a Quantum Example Oracle},
year = {1995},
publisher = {Association for Computing Machinery},
address = {New York, NY, USA},
url = {https://doi.org/10.1145/225298.225312},
doi = {10.1145/225298.225312},
booktitle = {Proceedings of the Eighth Annual Conference on Computational Learning Theory},
pages = {118–127},
numpages = {10},
series = {COLT 1995}
}

@article{valiant1984theory,
  title={A theory of the learnable},
  author={Valiant, Leslie G},
  journal={Communications of the ACM},
  volume={27},
  number={11},
  pages={1134--1142},
  year={1984},
  publisher={ACM New York, NY, USA}
}

@article{haussler_decision_1992,
	title = {Decision theoretic generalizations of the {PAC} model for neural net and other learning applications},
	volume = {100},
	issn = {0890-5401},
	url = {https://www.sciencedirect.com/science/article/pii/089054019290010D},
	doi = {https://doi.org/10.1016/0890-5401(92)90010-D},
	pages = {78--150},
	number = {1},
	journaltitle = {Information and Computation},
	author = {Haussler, David},
	date = {1992},
}

@article{PhysRevLett.98.160501,
  title = {Discriminating States: The Quantum Chernoff Bound},
  author = {Audenaert, K. M. R. and Calsamiglia, J. and Mu\~noz-Tapia, R. and Bagan, E. and Masanes, Ll. and Acin, A. and Verstraete, F.},
  journal = {Phys. Rev. Lett.},
  volume = {98},
  issue = {16},
  pages = {160501},
  numpages = {4},
  year = {2007},
  month = {4},
  publisher = {American Physical Society},
  doi = {10.1103/PhysRevLett.98.160501},
  url = {https://link.aps.org/doi/10.1103/PhysRevLett.98.160501}
}

@article{audenaert_asymptotic_2008,
	title = {Asymptotic Error Rates in Quantum Hypothesis Testing},
	volume = {279},
	% issn = {0010-3616, 1432-0916},
	doi = {10.1007/s00220-008-0417-5},
	abstract = {We consider the problem of discriminating between two different states of a ﬁnite quantum system in the setting of large numbers of copies, and ﬁnd a closed form expression for the asymptotic exponential rate at which the speciﬁed error probability tends to zero. This leads to the identiﬁcation of the quantum generalisation of the classical Chernoff distance, which is the corresponding quantity in classical symmetric hypothesis testing, thereby solving a long standing open problem.},
	pages = {251--283},
	number = {1},
	journaltitle = {Comm. Math. Phys.},
	shortjournal = {Commun. Math. Phys.},
	author = {Audenaert, K. M. R. and Nussbaum, M. and Szkola, A. and Verstraete, F.},
	urldate = {2023-06-21},
	date = {2008-04},
	langid = {english},
	% eprinttype = {arxiv},
	% eprint = {0708.4282 [quant-ph]},
	keywords = {Quantum Physics},
}

@article{caro_generalization_2022,
	title = {Generalization in quantum machine learning from few training data},
	volume = {13},
	url = {https://doi.org/10.1038/s41467-022-32550-3},
	doi = {10.1038/s41467-022-32550-3},
	pages = {4919},
	number = {1},
	journaltitle = {Nat. Commun.},
	shortjournal = {Nat. Commun.},
	author = {Caro, Matthias C. and Huang, Hsin-Yuan and Cerezo, M. and Sharma, Kunal and Sornborger, Andrew and Cincio, Lukasz and Coles, Patrick J.},
	date = {2022-08-22},
}

@misc{caro_information_theoretic_2023,
	title = {Information-theoretic generalization bounds for learning from quantum data},
	note = {{arXiv}:2311.05529},
	publisher = {{arXiv}},
	author = {Caro, Matthias and Gur, Tom and Rouzé, Cambyse and França, Daniel Stilck and Subramanian, Sathyawageeswar},
	urldate = {2023-11-10},
	date = {2023-11-09},
    year={2023},
	langid = {english},
	eprinttype = {arxiv},
	eprint = {2311.05529},
    PrimaryClass={quant-ph},

}

@book{hasminskii1981,
publisher = {Springer Science+Business Media, LLC},
series = {Stochastic Modelling and Applied Probability, 16},
title = {Statistical estimation: Asymptotic theory },
year = {1981},
author = {Ibragimov, I. A. and Has'minskii, R. Z. and Kotz, Samuel},
address = {New York},
booktitle = {Statistical estimation : asymptotic theory},
edition = {1st ed. 1981.},
isbn = {1-4899-0027-6},
language = {eng},
}

@article{horodecki_general_1999,
  title = {General teleportation channel, singlet fraction, and quasidistillation},
  author = {Horodecki, Micha\l{} and Horodecki, Pawe\l{} and Horodecki, Ryszard},
  journal = {Phys. Rev. A},
  volume = {60},
  issue = {3},
  pages = {1888--1898},
  numpages = {0},
  year = {1999},
  month = {9},
  publisher = {American Physical Society},
  doi = {10.1103/PhysRevA.60.1888},
}

@misc{holevo2005separability,
      title={Separability and Entanglement-Breaking in Infinite Dimensions}, 
      author={A. S. Holevo and M. E. Shirokov and R. F. Werner},
      year={2005},
      eprint={quant-ph/0504204},
      archivePrefix={arXiv},
      primaryClass={quant-ph}
}

@article{Regula_2021,
   title={Operational Quantification of Continuous-Variable Quantum Resources},
   volume={126},
   url={http://dx.doi.org/10.1103/PhysRevLett.126.110403},
   DOI={10.1103/physrevlett.126.110403},
   number={11},
   journal={Phys. Rev. Lett.},
   publisher={American Physical Society (APS)},
   author={Regula, Bartosz and Lami, Ludovico and Ferrari, Giovanni and Takagi, Ryuji},
   year={2021},
   month=mar }

@article{Haapasalo_2021,
   title={Operational Characterization of Infinite-Dimensional Quantum Resources},
   volume={127},
   url={http://dx.doi.org/10.1103/PhysRevLett.127.250401},
   DOI={10.1103/physrevlett.127.250401},
   number={25},
   journal={Phys. Rev. Lett.},
   publisher={American Physical Society (APS)},
   author={Haapasalo, Erkka and Kraft, Tristan and Pellonpää, Juha-Pekka and Uola, Roope},
   year={2021},
   month=12 }

@ARTICLE{integralkernel,
       author = {{Khrennikov}, A. Yu.},
        title = "{Two-particle wave function as an integral operator and the random field approach to quantum correlations}",
      journal = {Theoretical and Mathematical Physics},
         year = 2010,
        month = sep,
       volume = {164},
       number = {3},
        pages = {1156-1162},
          doi = {10.1007/s11232-010-0094-3},
       adsurl = {https://ui.adsabs.harvard.edu/abs/2010TMP...164.1156K}
}

\clearpage
\begin{appendices}


\section{Proof of Classical learning lower bound}\label{app:minimax}

For completeness, we reproduce the typical minimax bound that uses mutual information to lower bound the error of the best estimator on the worst-case distribution (e.g. Ref.~\cite{hasminskii1981}). 

\begin{repproposition}{prop:minimax} 
Let $\Ac$ be a bounded subset of a metric space $(X, d)$ and let random variable $\B$ be distributed according to $p_{B|\alpha}$ conditioned on $\alpha \in \Ac$. Define an estimator $D: \B \rightarrow \Ac$ for $\alpha \in \Ac$. For a non-decreasing $\ell: \R^+ \rightarrow \R$, we have
\begin{equation}
  \min_{D: \B \rightarrow \Ac} \max_{\alpha \in \Ac } \underset{{p_{B|\alpha}}}{\E} \left[ \ell(d(D(B), \alpha))\right] \geq \ell(\epsilon) \min_{\hat{V}: \B \rightarrow \mathcal{V}} \Pr_{p_V} \left( \hat{V}(B) \neq V\right).
\end{equation}
where the random variable $V\sim p_V$ is distributed uniformly over $\{1, \dots, |\mV|\}$, with $\mV:=\mathcal{V}(2\epsilon, \Ac, d)$ being $2\epsilon$-packing of $\Ac$ in $d$.
\end{repproposition}

\begin{proof}
    First, the total error is lower bounded by considering only the error arising from estimates that are at least $\epsilon$-far from $\alpha$, and assigning such estimates their minimal possible error of $\ell(\epsilon)$ (since $\ell$ is non-decreasing). Let $p_{\hat{\alpha}|\alpha}$ be the conditional distribution of $\hat{\alpha} = D(B)$ induced via the estimator $D: \B \rightarrow \A$ acting on $B\sim p_{B|\alpha}$, and henceforth denote $\hat{\alpha}:= D(B)$.  Defining $\mathbb{B}_\epsilon(\alpha) := \{\alpha' \in \Ac: d(\alpha, \alpha') < \epsilon\}$ and its complement $\mathbb{B}_\epsilon^c (\alpha) = \Ac \backslash \mathbb{B}_\epsilon(\alpha)$, then for any fixed $\alpha$ we have
\begin{align}
    \E_{p_{B|\alpha}} [ \ell(d(\hat{\alpha}, \alpha))] &= \int_{\Ac} d\hat{\alpha} p_{\hat{\alpha}|\alpha}(\hat{\alpha}|\alpha) \ell(d(\alpha, \hat{\alpha}))
    \\&= \int_{\mathbb{B}_\epsilon(\alpha)}  d\hat{\alpha} p_{\hat{\alpha}|\alpha}(\hat{\alpha}|\alpha) \ell(d(\alpha, \hat{\alpha})) + \int_{\mathbb{B}_\epsilon^c(\alpha)}  d\hat{\alpha} p_{\hat{\alpha}|\alpha}(\hat{\alpha}|\alpha) \ell(d(\alpha, \hat{\alpha})) 
    \\&\geq \ell(\epsilon)\int_{\mathbb{B}_\epsilon^c(\alpha)}  d\hat{\alpha} p_{\hat{\alpha}|\alpha}(\hat{\alpha}|\alpha) 
    \\&:= \ell(\epsilon)\Pr_{p_{\hat{\alpha}|\alpha}} (d (\alpha, \hat{\alpha}) \geq \epsilon)
\end{align} 
The first line uses several basic properties of the Markov chain $A \rightarrow B \rightarrow D(B)$:
\begin{align}
    \E_{p_{B|\alpha}} [ f(D(\beta), \alpha)] &= \int_\B d \beta p_{B|A}(\beta | \alpha) f(D(\beta), \alpha)
    \\&= \int_\B d \beta p_{B|A}(\beta | \alpha) f(D(B), \alpha) 
    \\&= \int_\B d \beta p_{B|A}(\beta | \alpha)  \underbrace{\int_\A d\hat{\alpha} p_{\hat{A}|B}(\hat{\alpha}|\beta) }_{=1} f(D(\beta), \alpha) 
    \\&= \underbrace{\int_\B d \beta p_{B|A}(\beta | \alpha)}_{=1}  \int_\A d\hat{\alpha} p_{\hat{A}|A}(\hat{\alpha}|\alpha)  f(\hat{\alpha}, \alpha) 
    \\&= \int_{\Ac} d\hat{\alpha} p_{\hat{\alpha}|\alpha}(\hat{\alpha}|\alpha) f(\hat{\alpha}, \alpha)
\end{align}
where we have applied the chain rule, conditional independence of $\hat{A}$ from $A$ given $B$, and conditional independence of $B$ from $\hat{A}$ given $A$. Consider $\hat{V}: \Ac \rightarrow \mathcal{V}(2\epsilon, \Ac, d)$ that returns the closest point in a $2\epsilon$-packing $\mathcal{V}(2\epsilon, \Ac, d)$ to an estimate $\hat{\alpha} := D(B) \in \Ac$. Writing $[|\mV|] := \{1, \dots, |\mV|\}$ and indexing every element $\alpha_v \in \mathcal{V}$ by $v \in [|\mV|]$, we define
\begin{equation}
    \hat{V}(\hat{\alpha}) = \argmin_{v \in S_{\mV}} d(\hat{\alpha}, \alpha_v).
\end{equation}
Since $[|\mV|]$ indexes a $2\epsilon$-packing, $d(\hat{\alpha}, \alpha_v) < \epsilon$ implies $\hat{V}(\hat{\alpha}) = v$ and conversely $\hat{V}(\hat{\alpha}) \neq v$ implies $d(\hat{\alpha}, \alpha_v) \geq \epsilon$. So, for any $v \in S_{\mV}$ we have
\begin{align}\label{eq:packing_ineq}
     \Pr_{p_{\hat{\alpha}|\alpha_v}}(\hat{V}(\hat{\alpha}) \neq v) \leq \Pr_{p_{\hat{\alpha}|\alpha_v}} (d(\hat{\alpha}, \alpha_v) \geq \epsilon)
\end{align}
Considering the most difficult element $\alpha \in \Ac$ to learn leads to the following chain of inequalities, which holds for any prior probability distribution $p_V$ on $[|\mV|]$:
\begin{align}
    \max_{\alpha \in \Ac} \Pr_{p_{\hat{\alpha}|\alpha}} (d(\hat{\alpha}, \alpha) \geq \epsilon) &\geq \max_{v \in [|\mV|]} \Pr_{p_{\hat{\alpha}|\alpha_v}} (d(\hat{\alpha}, \alpha_v) \geq \epsilon)
    \\&\geq \E_{p_V} \Pr_{p_{\hat{\alpha}|\alpha_v}} (d(\hat{\alpha}, \alpha_v) \geq \epsilon)
    \\&= \sum_{v \in [|\mV|]} \Pr_{p_{\hat{\alpha}|\alpha_v}} (d(\hat{\alpha}, \alpha_v) \geq \epsilon) p_V(v)
    \\&\geq \sum_{v \in [|\mV|]} \Pr_{p_{\hat{\alpha}|\alpha_v}}( \hat{V}(\hat{\alpha}) \neq v) p_V(v)\label{eq:line23}
    \\&= \Pr_{p_V}( \hat{V}(\hat{\alpha}) \neq V)
\end{align}
Combining Eqs.~\ref{eq:packing_ineq} and~\ref{eq:line23} and minimizing over estimates $\hat{\alpha}$ gives
\begin{align}
     \min_{\hat{\alpha} } \max_{\alpha \in \Ac } \underset{{p_{\hat{\alpha}|\alpha}}}{\E} \left[ \ell(d(\hat{\alpha}(B), \alpha))\right] &\geq \min_{\hat{\alpha} } \max_{\alpha \in \Ac } \ell(\epsilon)  \Pr_{p_{\hat{\alpha}|\alpha}} (d(\hat{\alpha}, \alpha) \geq \epsilon)
     \\&\geq  \min_{\hat{\alpha} } \ell(\epsilon) \Pr_{p_V}( \hat{V}(\hat{\alpha}) \neq V)
     \\&:=  \ell(\epsilon) \min_{\hat{V}: \mathcal{B} \rightarrow [|\mV|] } \Pr_{p_V}( \hat{V}(B) \neq V)
\end{align}
where we have observed that the optimal discrete estimator $\hat{V}$ acting on observations in $\mathcal{B}$ performs equally well as the the optimal discretized estimator $\hat{\alpha}$ applied to those same observations.

\end{proof}

\section{Proof of classical learning upper bound}\label{app:maximax}

We prove the following statement from the main text:
\begin{repproposition}{prop:maximax}
Let $\Ac$ be a bounded subset of a metric space $(X, d)$ and let random variable $\B$ be distributed according to $p_{B|\alpha}$ conditioned on $\alpha \in \Ac$. Then, for non-increasing, positive $\score: \R^+ \rightarrow \R^+$, the best-case performance (with respect to $\alpha \in \Ac$) of an optimal estimator $\hat{\alpha}$ is lower bounded according to 
    \begin{equation}
    \max_{\alpha \in \mathrm{A}} \max_{D: \mathcal{B}\rightarrow \Ac} \E_{p_{B|\alpha}} [\score(d(\alpha, D(B)))] \geq \score(\epsilon) 2^{-\HHH(W|B)}
\end{equation}
where $W$ is a random variable uniformly distributed on the set $\{1, \dots, |\mathcal{W}\left(\epsilon, \Ac, d\right)|\}$.
\end{repproposition}
\begin{proof}

First, for any fixed $\alpha \in \Ac$ the learner's expected score for sampling $\hat{\alpha}:= \hat{\alpha}(B)$ with respect to the induced distribution $p_{\hat{\alpha}|\alpha}$ is no less than the minimum expected score considering only estimates $\hat{\alpha}$ that are $\epsilon$-close to $\alpha$. Then, we have
\begin{align}
    \E_{p_{\hat{\alpha}|\alpha}} [\score(d(\alpha, \hat{\alpha}))] &=\E_{p_{\hat{\alpha}|\alpha}} [\score(d(\alpha, \hat{\alpha})) ( \I\{d(\alpha, \hat{\alpha}) \leq \epsilon\} + \I\{d(\alpha, \hat{\alpha}) > \epsilon\})]
    \\&\geq \E_{p_{\hat{\alpha}|\alpha}} [\score(d(\alpha, \hat{\alpha})) \I\{d(\alpha, \hat{\alpha}) \leq \epsilon\}]\label{eq:l17}
    \\&\geq \score(\epsilon) \Pr_{p_{\hat{\alpha}|\alpha}} (d(\alpha, \hat{\alpha}) \leq \epsilon).  \label{eq:l18}
\end{align}
Line~\ref{eq:l17} requires $\score(t) \geq 0$ while line~\ref{eq:l18} used the fact that $\score$ is non-increasing. Defining an $\epsilon$-net $\mathcal{W}\left(\epsilon, \Ac, d\right)$ with elements $\alpha_w$ indexed by $w \in [|\mW|]:=\{1, \dots, |\mathcal{W}\left(\epsilon, \Ac, d\right)|\}$ we have
\begin{align}
    \max_\alpha \Pr_{p_{\hat{\alpha}|\alpha}} (d(\alpha, \hat{\alpha}) \leq \epsilon) &\geq \max_{w \in [|\mW|]} \Pr_{p_{\hat{\alpha}|\alpha_w}} (d(\alpha_w, \hat{\alpha}) \leq \epsilon)
    \\&\geq \underset{w \in [|\mW|] }{\E}\Pr_{p_{\hat{\alpha}|\alpha_w}} (d(\alpha_w, \hat{\alpha}) \leq \epsilon ) \label{line:rgd}
\end{align}
The first inequality follows from $\mathcal{W}(\epsilon, \Ac, d)\subset \Ac$,and the second follows since the best-case estimator always outperforms the average-case estimator. For some distribution on $[|\mW|]$, the RHS of Eq.~\ref{line:rgd} might result in a trivial bound since some distributions on $[|\mW|]$ will not contain any $w$ such that $d(\alpha_w, \hat{\alpha}) \leq \epsilon$. We may account for this possibility by evaluating the final expression for the uniform distribution on $\mathcal{W} $, i.e. 
\begin{equation}\label{line:c2}
    \max_\alpha \Pr(d(\alpha, \hat{\alpha}) \leq \epsilon) \geq \frac{1}{|\mathcal{W}|}\sum_{w \in [|\mW|]} \Pr_{p_{\hat{\alpha}|\alpha_w}} (d(\alpha_w, \hat{\alpha}) \leq \epsilon ).
\end{equation}
Finally, we substitute the continuous estimator for $\hat{\alpha}$ for a discrete estimator $\hat{W}: \Ac \rightarrow \mathcal{W}\left(\epsilon, \Ac, d\right)$ to recover a multi-hypothesis testing scenario. The discrete estimator is given by any $w \in [|\mW|]$ satisfying
\begin{equation}
    \hat{W}(\hat{\alpha}) = \argmin_{w \in [|\mW|]} d(\hat{\alpha}, \alpha_w)
\end{equation}
For some fixed $\alpha$, suppose $w$ is an index for which $d(\alpha, \alpha_w)\leq \epsilon$ (such an index must exist since we are working with an $\epsilon$-net). Then, $\hat{W} = w$ implies $d(\hat{\alpha}, \alpha_w) \leq \epsilon$. This is opposite to the implication in Eq.~\ref{eq:packing_ineq}, due to the possibility of overlap between the $\epsilon$-balls composing a covering, so that several elements of $\mathcal{W}(\epsilon, \Ac, d)$ may be $\epsilon$-close to $\hat{\alpha}$. Intuitively, the optimal estimator $\hat{W}$ has a lower probability of success because, for some fixed $\alpha_w$, $\hat{W}$ must output a unique solution from the $\epsilon$-covering partition, whereas several elements of $\mW$ may be $\epsilon$-close to $\alpha_w$. For any random variable $W$ taking values in $[|\mW|]$, the optimal estimator then obeys 
\begin{align}\label{eq:epscovei_Rneq}
    \max_{\hat{\alpha}}\Pr(d(\hat{\alpha}, \alpha_w) \leq \epsilon) \geq \max_{\hat{W}: \Ac \rightarrow [|\mW|]}\Pr(\hat{W}(\hat{\alpha}) = w).
\end{align}
Finally, from the definition of conditional min-entropy we have  
\begin{equation}\label{line:c3}
    -\log \max_{\hat{W}} \Pr(\hat{W} = W) = \Hmin(W|B) \leq \HHH(W|B).
\end{equation}
Combining Eqs.~\ref{line:rgd},~\ref{line:c2},~\ref{eq:epscovei_Rneq}, and~\ref{line:c3}, and applying the appropriate maximizations and logarithm (imposing strict positivity on $\score$ if necessary), we recover the statement from the main text.
\end{proof}

\section{Proof of entanglement fraction upper bound}\label{app:ent_recovery}

We will now prove the entanglement fraction guarantee presented in Theorem~\ref{thm:ent_recovery}. It is helpful to refer to the analogous classical learning guarantee: Recall that Proposition~\ref{prop:maximax} was proven by covering the space $\Ac$ with balls having radius no greater than $\epsilon$, and then showing that predicting the target parameter $\alpha$ to $\epsilon$ accuracy is easier than succeeding at a multi-hypothesis testing task for determining which ball has the closest center to the target parameter $\alpha$. The latter task may be understood as determining membership of $\alpha$ to an element of an $\epsilon$-covering partition of $\Ac$. Analogously, the technique here involves partitioning of a state space into a finite number of regions, and then working with classes of states defined piece-wise with respect to these regions. In this case, optimizing singlet fraction with respect to this class of states takes on a relationship to the entanglement fraction task that is similar to the relationship between hypothesis testing and classical learning. 

We return to the discussion of quantum channels, following the treatment of Ref.~\cite[Chapter 6]{attal2014quantum}. For separable Hilbert spaces $\X, \Y$, the partial trace $\tr_{\Y}: \X\otimes \Y \rightarrow \X$ is the unique operator satisfying $\tr(\tr_{\mathcal{Y}}(A) X) = \tr(A(X\otimes \I))$ for all $X \in \mB(\X)$ and $A \in \mT(\X \otimes \Y)$. For a linear map $\mN: \mB(\mH) \rightarrow \mB(\mH')$, then $\mN$ is a channel, denoted $\mN \in \cptp(\mH, \mH')$, if and only if for all $X \in \mT(\mH)$, we have $\mN(\rho) = \tr_{\mathcal{E}}(U(X \otimes \Sigma)U^\dagger)$ for some separable Hilbert space $\mathcal{E}$, some (pure) state $\Sigma \in \D(\mathcal{E})$ , and unitary operator $U \in \U{\mH \otimes \mathcal{E}, \mH' \otimes \mathcal{E}}$. An alternative characterization of quantum channels may be given in terms of completely positive unital linear maps. An operator $\mathcal{M}: \mB(\mH') \to \mB(\mH)$ is positive if it sends positive operators to positive operators, and is completely positive if the extension $\mathcal{M}_n$ of $\mathcal{M}$ to $\mB(\mH \otimes \C^n)$ given by $\mathcal{M}_n(X \otimes Y) = \mathcal{M}(X) \otimes Y$ is positive. We say that the operator is unital if $\mathcal{M}(\I) = \I$. Then, $\mN$ is a channel if there is a completely positive, unital $\mN^*: \mB(\mH') \rightarrow \mB(\mH)$ such that $\tr(\mN(X)Y) = \tr(X \mN^*(Y))$ for all $X \in \mT(\mH)$ and $Y \in \mB(\mH)$. We omit any further characterization of these maps, including the celebrated results of Stinespring \cite{stinespring} and Kraus \cite{kraus}, and direct the interested reader to Ref.~\cite{attal2014quantum}.

We once again consider the compact subset of a metric space $\Ac \subset \R^p$ representing the space of parameters $\alpha$ that a learner attempts to estimate in the classical setting. Instead of a probability distribution over $\Ac$, we will consider continuous-variable states defined with respect to this parameter space. We will consider pure states as square-integrable functionals $\Ac \rightarrow \mathbb{C}$ defined at each point in $\Ac$, thereby identifying each state $|\Psi\rangle \in \mH$ with $\Psi \in L^2(\Ac)$ and fixing our Hilbert space to be $\mH \sim L^2(\Ac)$. We define a partition of the subset $\Ac$ as a set $\{\Ac_1, \dots, \Ac_N\}$ of subsets that are non-intersecting and whose union recovers $\Ac$: $\bigcup_i \Ac_i = \Ac$ and $\Ac_i \cap \Ac_j = \emptyset$ whenever $i \neq j$, for $i=1, \dots, N$. For each subset in the partition, we define a space of states $\mH^{(i)}$ containing all square-integrable functions having support only on $\Ac_i$:
\begin{equation}
    \mH^{(i)} \sim L^2_{i}(\Ac_i) := \{\Psi \in L^2(\Ac): \Psi(\alpha) = 0 \text{ for all } \alpha \notin \Ac_i\},
\end{equation}
where $L^2_{i}$ may be thought of a $L^2$ with respect to a uniform measure having support only on $\Ac_i$. Each $L^2_i(\Ac_i)$ is embedded in $L^2(\Ac)$ in a straightforward way, such that $\langle \Psi, \Phi \rangle = 0$  for $\Psi \in \mH^{(i)}, \Phi \in \mH^{(j)}$ whenever $i \neq j$ (the inner product is computed with respect to the measure $\I\{\alpha \in \Ac_i\}$ on $\R^p$). As such, we write\footnote{For our setting it does not matter whether each subset $\Ac_i$ in the partition of $\Ac$ is open and/or closed. For a sense of closure, the reader may consider $L^2_i(\Ac_i)$ to be defined with respect to the \textit{closure} of $\Ac_i$. This does not affect any of our statements, as $L^2_i(\Ac_i)$ forms an equivalence class of functions on $\Ac\subset \R^p$ up to a (well-behaved) intersection of the closures of $\Ac_i$ and $\Ac_j$ having dimension at most $p-1$.}
\begin{equation}
    \bigoplus_{i=1}^N \mH^{(i)} = \mH, \qquad \mH^{(i)} \perp \mH^{(j)} \text{ for } i \neq j.
\end{equation}
We will define a set of pure states in $\mH$ that are piecewise constant within a set of subspaces $P = \{\mH^{(1)}, \dots, \mH^{(N)}\}$. This forms a finite-dimensional subspace of $\mH$ which we define as 
\begin{equation}\label{eq:partitionH}
    \mH_P = \{|\Psi\rangle \in \mH, \, \Psi(x) = \Psi(x') \text{ for all } x, x' \in \Ac_i, i=1,\dots N\}
\end{equation}
Letting $|i\rangle \in \mH$ be constant on $\mH^{(i)}$ and zero elsewhere, every element of $\mH_P$ may be written:
\begin{align}
    |\Psi\rangle &= \sum_{i=1}^N c_i |i_P\rangle.
\end{align}
By identifying $|i_P\rangle$ with a piecewise constant function on $L^2_i(\Ac)$, it follows that the expression on the right is bounded and normalizable, hence $\dm{\Psi}$ is a valid state. To construct bounds for the entanglement fraction task, we will need to reason about finite-dimensional subspaces of infinite-dimensional $\mH$. For this, the following will be useful.
\begin{lemma}\label{prop:cptp}
    Let $\mH, \mH'$ be separable Hilbert spaces and define  finite-dimensional subspaces $\mK \subset \mH$, $\mK' \subset \mH'$. Let $\Pi_{\mK}: \mH \rightarrow \mH$ be a projector onto the subspace $\mK$ of an orthogonal decomposition $\mH = \mK \oplus \mK^\perp$. Let $\mathcal{M} \in \cptp(\mK, \mK')$ be a fixed channel. Then, for any $X_H \in \mT(\mH)$ such that $\Pi_{\mK} X_H \Pi_{\mK} := X_K \in \mT(\mK)$, there exists some $\mN \in \cptp(\mH, \mH')$ such that
    \begin{equation}
        \left. \mN(X_H) \right|_{\mK} = \mathcal{M}(X_K)
    \end{equation}
\end{lemma}
\begin{proof}
    We invoke the unitary representation of a channel $\mathcal{M}$. Assuming for now that $\dim(\mK) = \dim(\mK')$, there exists some choice of separable Hilbert space $\mathcal{E}$, unitary $U \in \U{\mK \otimes \mathcal{E}}$, and $\Omega \in \D(\mathcal{E})$ such that 
    \begin{equation}
        \mathcal{M}(X_K) = \tr_{\mathcal{E}} ( U (X_K \otimes \Omega) U^\dagger ),
    \end{equation}
     Then we choose an embedding $V \in \U{\mH \otimes \mathcal{E}}$ of $U$, e.g. $V = U \oplus W$ for some $W \in \U{\mK^\perp \otimes \mathcal{E}}$. Then, we find
    \begin{equation}
        \tr_{\mathcal{E}}(V (X_H \otimes \Omega) V^\dagger) := \mN(X_H)
    \end{equation}
     defines a valid quantum channel in $\cptp(\mH)$, and by construction $\mN(X_H) = \mathcal{M}(X_K)$. If $\dim(\mK) \neq \dim(\mK')$, we instead choose a separable Hilbert space $\mathcal{E}'$ such that $U$ is a unitary map from $\mK \otimes \mathcal{E} \to \mK' \otimes \mathcal{E}'$. The desired result follows instead for $\mathcal{M}(X_K)~=~\tr_{\mathcal{E}'}(U (X_K \otimes \Omega) U^\dagger)$.
\end{proof}
Recall that for an $\epsilon$-net we define  $\hat{W}: \Ac \rightarrow \Ac$ to be the function that returns the index of some closest point in $\mW(\epsilon, \Ac, d)$ to $r$:
\begin{equation}
    \hat{W}(r) \in \{w: \alpha_w = \argmin_{\alpha_{w'} \in \mW} d(r, \alpha_{w'})\},
\end{equation}
Since $\mW$ is finite a minimimum exists and any ties are broken in an arbitrary but consistent way.\footnote{For our analysis it should be that the set of possible values for $\hat{W}(r)$ is measure zero in $\Ac$ with respect to the Lebesgue measure anyways.} Given the above preliminaries, we may state our desired result:
\begin{reptheorem}{thm:ent_recovery}[Entanglement fraction guarantee]
     Given separable Hilbert spaces $\mH_R, \mH_A, \mH_B, \mH_{\hat{A}} \simeq L^2(\Ac)$ and a fixed $\mathcal{N} \in \cptp(\mH_A, \mH_B)$, define the (unnormalized) state $|\Phi_\epsilon\rangle_{R\hat{A}} \in \mH_{R \hat{A}} $ and $|\Psi\rangle_{RA}\in \mH_{RA}$ according to:
    \begin{align}
        |\Phi_\epsilon\rangle_{R\hat{A}} &= \int_{\Ac} dr \int_{\mathbb{B}_\epsilon(r)} d\hat{\alpha} |r\rangle_R \otimes  |\hat{\alpha}\rangle_{\hat{A}} \\
        |\Psi\rangle_{RA} &=  \int_{\Ac} dr  \int_{\Ac} d\alpha  \Psi_{RA}(r, \alpha) |r\rangle_R \otimes | \alpha \rangle_{A}
    \end{align}
    where $\mathbb{B}_\epsilon(r) := \{r'\in \Ac: d(r, r')\leq \epsilon\} \cap \Ac$ is the $\epsilon$-ball centered at $r$. Then, letting $\Sigma_{RB} = (\I\otimes \mN)(\dm{\Psi}_{RA})$, we have
    \begin{equation}
         \sup_{|\Psi\rangle  \in \mH_{RA}} \sup_{\mathcal{D} \in \cptp(\mH_B, \mH_{\Ahat})} \log \langle \Phi |\Pi_{\Wc^2} (\I  \otimes \mathcal{D}) (\Sigma_{RB}) \Pi_{\Wc^2}| \Phi\rangle_{R\hat{A}} \geq  -  \HHH(R|B)_\rho   
    \end{equation}
    where $\Pi_{\Wc^2}$ is the projector onto the subset $\{(r, \alpha) \in \Ac \times \Ac : \hat{W}(r) = \hat{W}(\alpha)\}$, and $\HHH(R|B)_\rho$ is computed with respect to $\rho_{RB} = (\I \otimes \mathcal{N})(\dm{\phi}_{RA})$, the outcome of applying $\mN$ to one half of a maximally entangled state $|\phi\rangle \in \mathbb{C}^{|\mathcal{W}|} \otimes \mathbb{C}^{|\mathcal{W}|}$.

\end{reptheorem}
\begin{proof}
The claim is intuitive as a statement about a quantum communication protocol. Consider two parties, Alice (with access to infinite dimensional systems $RA$) and Bob (infinite dimensional system $B$). Alice wishes to transmit quantum entanglement through a channel $\mN\in\cptp(\mH_A, \mH_B)$, in the sense that Bob will apply a channel $\mD \in \cptp(\mH_B, \mH_{\Ahat})$ in order to maximize the overlap of the final system $R\Ahat$ with some entangled state $|\Phi_\epsilon\rangle$. Then, given a full characterization of $\mN$, the two parties agree to use finite dimensional subspaces -- Alice prepares states with finite Schmidt rank, Bob considers channels with nontrivial action on a  subspace of input states -- and they maximize overlap of the final system with some projection of $|\Phi_\epsilon\rangle$. The bound characterizes optimal performance for this sequence of (embedded) finite-dimensional operations.

For each $w \in [|\mW|]$ we define $\Wc(w) := \{r: r\in \Ac, \hat{W}(r) = w\} \subset \Ac$ as the set of points in $\Ac$ for which the closest element of $\mW$ has index $w$. By construction, $\bigcup_{w \in [|\mW|]} \Wc(w) = \Ac$ and $\Wc(w) \cap \Wc(v) = 0$ for $w \neq v$, i.e. $\Wc:=\{\Wc(w): w \in [|\mW|]\}$ is an $\epsilon$-covering partition of $\Ac$ (c.f. Definition~\ref{def:ecp} of the main text). It will be convenient to use notation for describing elements of $L^2(\Ac)$ in terms of the linear forms $\{|x\rangle: x \in \mathcal{X}\subseteq \Ac\}$, e.g.
\begin{equation}
    |w\rangle:= \int_{\Wc(w)} dx |x\rangle.
\end{equation}
In particular, we are interested in two subspaces:
\begin{enumerate}
    \item $\mH_{\Wc} \subset \mH_{\Ahat}$, defined as the subspace of states which are piecewise constant on the $\epsilon$-covering partition $\Wc$ of $\Ac$ (see Eq.~\ref{eq:partitionH}). These states have the form
    \begin{equation} 
    |\psi\rangle_{\Ahat} = \sum_{w=1}^{|\mW|} \psi_w |w\rangle
    \end{equation}
    \item $\mH_{\Wc^2} \subset \mH_{RA}$, defined as the subspace of bipartite states of the form 
    \begin{equation}\label{eq:psi_ra_def}
        \ket{\psi}_{RA} = \sum_{w=1}^{|\mW|} \psi_w \ket{w}_R \otimes \ket{w}_A.
    \end{equation}
\end{enumerate}
Given that $\mH_{\Wc} \simeq \C^{|\mW|}$ and is therefore closed, we may partition $\mH_{\Ahat}$ into orthogonal subspaces, $\mH_{\Ahat}= \mH_{\Wc} \oplus \mH_{\Wc}^\perp$. Letting $\Pi_{\Wc}$ and $\Pi_{\Wc}^\perp$ be orthogonal projectors onto these respective subspaces, we have that $\Pi_{\Wc} |\psi\rangle = |\psi\rangle$ for $|\psi\rangle \in \mH_{\Wc}$. We will similarly define a projector that preserves states in $\mH_{\Wc^2}$. As $\Ac$ is a compact subset of $\R$, we may partition $\Ac \times \Ac$ into a finite number of bounded regions. We  consider the subset  of $\Ac \times \Ac$ containing coordinate pairs that share a closest point in an $\epsilon$-net $\mW(\epsilon, \Ac, d):=\mW$. We call this subset $K \subset \Ac \times \Ac$, defined as
\begin{align}
    K := \{(r, \alpha) \in \Ac \times \Ac : \hat{W}(r) = \hat{W}(\alpha)\}.
\end{align}
Define a projector $\Pi_{\Wc^2}$ onto the subspace of $\mH_R \otimes \mH_{\Ahat}$ spanned by linear forms $|r\rangle \otimes |\alpha\rangle$ with $(r, \alpha) \in K$, i.e. all $r, \alpha$ that share a nearest point in $\mathcal{W}$. $\Pi_{\Wc^2}$ may be written in this way as
\begin{align}
    \Pi_{\Wc^2} &= \int dr \int d\alpha \I \{\hat{W}(r) = \hat{W}(\alpha)\} \dm{r}\otimes \dm{\alpha}
    \\&= \int dr \int d\alpha \sum_{w=1}^{|\mW|}\I \{\hat{W}(r) = \hat{W}(\alpha) = w\} \dm{r}\otimes \dm{\alpha}
    \\&= \int dr \int d\alpha \sum_{w=1}^{|\mW|}\I \{\hat{W}(r) = w \} \I\{\hat{W}(\alpha) = w\} \dm{r}\otimes \dm{\alpha}
    \\&=  \sum_{w=1}^{|\mW|}\int dr \int d\alpha\I \{\hat{W}(r) = w \} \I\{\hat{W}(\alpha) = w\} \dm{r}\otimes \dm{\alpha}
    \\&= \sum_{w=1}^{|\mW|}\int\limits_{\Wc(w)} dr \int\limits_{\Wc(w)} d\alpha \dm{r}\otimes \dm{\alpha}
\end{align}
We then find,
\begin{align}
    \Pi_{\Wc^2}|\Phi_\epsilon\rangle  &= \sum_{w=1}^{|\mW|}   \int\limits_{\Wc(w)} dr \int\limits_{\Wc(w)} d\alpha \int\limits_{\mathbb{B}_\epsilon(r)} d\alpha' |r\rangle \delta(\alpha - \alpha') |\alpha\rangle
    \\&= \sum_{w=1}^{|\mW|}   \int\limits_{\Wc(w)} dr \int\limits_{\Wc(w)} d\alpha   |r\rangle  |\alpha\rangle
    \\&=  \sum_{w=1}^{|\mW|} |w\rangle \otimes |w\rangle
    \\&:= |\mW|^{1/2}  |\phi\rangle_{R\Ahat}, 
\end{align}
which follows since $\Wc(w) \cap\mathbb{B}_\epsilon(r) =\emptyset$ whenever $\hat{W}(r) \neq w$, and 
\begin{equation}\label{eq:littlephi}
    |\phi\rangle_{RA}:= \frac{1}{|\mW|^{1/2}} \sum_{w=1}^{|\mW|} |w\rangle \otimes |w\rangle
\end{equation}
is a maximally entangled state in $\mH_{\Wc^2}$. With these operators in place, we find
\begin{align}
     \sup_{|\Psi\rangle \in \mH_{RA}} &\sup_{\mathcal{D} \in \cptp(\mH_B, \mH_{\Ahat})} \langle \Phi_\epsilon | \Pi_{\Wc^2} (\I\otimes \mathcal{D}\circ \mathcal{N} )(\dm{\Psi}) \Pi_{\Wc^2} |\Phi_\epsilon \rangle 
     \\&= \sup_{|\Psi\rangle \in \mH_{RA}} \sup_{\mathcal{D} \in \cptp(\mH_B, \mH_{\Ahat})} |\mW| \,\langle \phi| (\I\otimes \mathcal{D}\circ \mathcal{N} )(\dm{\Psi}) |\phi \rangle_{R\Ahat}\label{line:a1}
     \\&\geq \sup_{|\psi\rangle \in \mH_{\Wc^2}} \sup_{\mathcal{D} \in \cptp(\mH_B, \mH_{\Ahat})} |\mW| \, \langle \phi| (\I\otimes \mathcal{D}\circ \mathcal{N} )(\dm{\psi}) |\phi \rangle_{R\Ahat} \label{line:a2}
     \\&\geq \sup_{|\psi\rangle \in \mH_{\Wc^2}} \sup_{\mathcal{D} \in \cptp(\mH_{\Wc})} |\mW|\,\langle \phi| (\I\otimes \mathcal{D}\circ \mathcal{N} )(\dm{\psi}) |\phi \rangle_{R\Ahat} \label{line:a3}
\end{align}
In Line~\ref{line:a1} we have substituted Eq.~\ref{eq:littlephi}, and Line~\ref{line:a2} follows since $\mH_{\Wc^2} \subset \mH_{RA}$. For Line~\ref{line:a3} we apply Lemma~\ref{prop:cptp}. Since $|\psi\rangle \in \mH_{\Wc^2}$ is finite-dimensional, we may consider a restriction of $\mN$ to $\mT(\mH_{\Wc})$ for some finite dimensional subspace $\mH_{\Wc}^A \subset \mH_A$, the image of which also lies in some finite-dimensional subspace (e.g. $\text{span}\{\mN(|i\rangle \langle j|)\} \subset \mT(\mH_B)$, where $\{|i\rangle\} \subset \mH_{\Wc}^A$ are an orthonormal basis), and this image has dimension no greater than $|\mW|^2$. We then identify some $\mH_{\Wc}^B \subset \mH_B$ with dimension $\dim \mH_{\Wc}^B = \dim \mH_{\Wc} = |\mW|$ such that $\mN(\mT(\mH_{\Wc}^A))$ is contained in $\mT(\mH_{\Wc}^B)$. In other words, for $\mH_B = \mH_{\Wc}^B \oplus (\mH_{\Wc}^B)^\perp$ and $X \in \mT(\mH_{\Wc}^A)$, and we can embed $\mN(X) \in \mT(\mH_{\Wc}^B)$ into $X_B \in \mT(\mH_{B})$ such that $\mN(X)  = \Pi_{\Wc'} X_B \Pi_{\Wc'}$ for an orthogonal projector $\Pi_{\Wc'}$ onto $\mH_{\Wc}^B$. Finally, choosing $\mH_{\Wc}^{\Ahat} \subset \mH_{\Ahat}$ with $\dim \mH_{\Wc}^{\Ahat} = \dim \mH_{\Wc}$ and $\mH_{\Ahat} = \mH_{\Wc}^{\Ahat} \oplus (\mH_{\Wc}^{\Ahat})^\perp$, by Lemma~\ref{prop:cptp} for any channel $\mD' \in \cptp(\mH_{\Wc}^B, \mH_{\Wc}^{\Ahat})$ which has an extension $\mD \in \cptp(\mH_B, \mH_{\Ahat})$ such that $\mD'(\mN(X)) = \mD(X_B)$, which implies the final inequality. Since $\dim\mH_{\Wc}^A = \dim\mH_{\Wc}^B = \dim\mH_{\Wc}^{\Ahat}$, all of these subspaces are isomorphic to $\mH_{\Wc} = \C^{|\Wc|}$ and we drop the superscripts for clarity. However, it is implicit that the finite-dimensional channels $\mD$ are embedded in such a way to act only on finite-dimensional subsystems of infinite-dimensional states.

We now apply the definition of conditional min-entropy (Eq.~\ref{eq:minent}) for channels $\mD \in \cptp(\mH_{\Wc})$, we find
\begin{align}
     \sup_{|\Psi\rangle \in \mH_{RA}} \sup_{\mathcal{D} \in \cptp(\mH_B, \mH_{\Ahat})} \langle \Phi_\epsilon | \Pi_{\Wc^2} (\I\otimes \mathcal{D}\circ \mathcal{N} )(\dm{\Psi}) \Pi_{\Wc^2} |\Phi_\epsilon \rangle 
     &= \sup_{|\psi\rangle \in \mH_{\Wc^2}} \left(\sup_{\mathcal{D} \in \cptp(\mH_{\Wc})} |\mW|\, \langle \phi| (\I\otimes \mathcal{D} )(\rho_{RB}) |\phi \rangle_{R\Ahat} \right)
     \\&= \sup_{|\psi\rangle \in \mH_{\Wc^2}} - \Hmin(R|B)_{\rho}
     \\&\geq \sup_{|\psi\rangle \in \mH_{\Wc^2}} - \HHH(R|B)_{\rho}
\end{align}
with both entropies being taken with respect to $\rho_{RB} = (\I \otimes \mN)(\dm{\psi})$. This supremum is lower bounded by choosing any state in $\mH_{\Wc^2}$, and in particular we choose $|\phi\rangle$ to recover the statement of Theorem~\ref{thm:ent_recovery}.
\end{proof}

\section{Proof of entanglement fraction lower bound}\label{app:ent_bound}

Here, we prove the lower bound on the error for the entanglement fraction task (Theorem~\ref{thm:ent_bound}). As mentioned already, the minimum of $ \langle \Phi_\epsilon | (\I\otimes \mathcal{D}\circ \mathcal{N} )(\Psi) |  \Phi_\epsilon \rangle$ with respect to  $|\Psi\rangle \in \mH_{RA}$ will be trivially small, since the system $RA$ could always be prepared in a separable state. Instead, we will consider a constrained minimization in which the system $RA$ is guaranteed to have some fixed amount of entanglement, and then ask what the optimal performance of a quantum learner is for the worst-case initial state? We will now derive a bound for this task in terms of the conditional Von Neumann entropy of the system $RB$ using the quantum Fano inequality.

To this end, we will define a specific entangled state in infinite dimensions, and then impose that $|\Psi\rangle_{RA}$ has some minimum overlap with that state. Consider two  separable Hilbert spaces $\mH_R$ and $\mH_A$ with orthonormal bases $\{|j_R\rangle\}$ and $\{ |j_A\rangle\}$ respectively. For any fixed $k \in \mathbb{N}$, we consider an orthonormal set 
\begin{equation}
    S = \{ |i_R\rangle \otimes |j_A\rangle: i \in \mathcal{I}, j \in \mathcal{J} \}
\end{equation}
which is specified by sets of indices $\mathcal{I},\mathcal{J} $ of size $|\mathcal{I}|=|\mathcal{J}|=k$ and spans a $k^2$-dimensional subspace of $\mH_{RA}$. In this subspace there is a maximally entangled state $|\Omega^{S}\rangle \in \mH_{RA}$ of the form
\begin{equation}
    |\Omega^{S}\rangle := \frac{1}{\sqrt{k}} \sum_{j=1}^k |j_R\rangle \otimes |j_A\rangle.
\end{equation}
The maximal singlet overlap of any $\Psi \in \mH_{RA}$ with some maximally entangled state in $\text{span}(S)$ is then defined as 
\begin{align}
    q_k(\Psi, S) = \sup_{\mathcal{M} \in \cptp(\mH_A)} \bigl\langle \Omega^S| ( \I \otimes \mathcal{M})(\Psi) | \Omega^{S} \bigr\rangle 
\end{align}
From this, we may define a notion of singlet fraction for $\Psi$ as the largest achievable singlet fraction, with respect to all subsets $S$:
\begin{equation}
    q_k(\Psi) := \sup_{\substack{S\subset \mH_{RA} \\ |S| = k^2}} q(\Psi, S),
\end{equation}
where the supremum is to be taken over choices of bases for $\mH_R$ and $\mH_A$ as well as sets $\mathcal{I}$ and $\mathcal{J}$. In this way, $q(\Psi)$ describes the maximum overlap of $\Psi$ with \textit{some} maximally entangled state in any $k\times k$-dimensional subspace of $\mH_{RA}$. 

To prove Theorem~\ref{thm:ent_bound}, we follow a similar procedure as for the proof of Theorem~\ref{thm:ent_recovery}, namely we focus on states that are piecewise constant within the $\epsilon$-balls that constitute an $\epsilon$-packing $\mV:= \mV(\epsilon, \Ac, d)$. We begin by defining the function $\hat{V}: \Ac \rightarrow \Ac$ that identifies the $\epsilon$-ball containing $r \in \Ac$:
\begin{equation}
    \hat{V}(r) \in \{v: \alpha_v = \argmin_{\alpha_{v'} \in \mV} d(r, \alpha_{v'})\},
\end{equation}
and as before, ties are irrelevant and may be broken in some consistent way. For each $\alpha_v \in \mV$,  we define $\Vc(v) := \mathbb{B}_\epsilon(\alpha_v)$ the $\epsilon$-ball at $\alpha_v$. Unlike the scenario of Theorem~\ref{thm:ent_recovery}, the set of $\epsilon$-balls composing $\mathcal{V}$ does not form a partition of $\Ac$. An (unnormalized) state that is piecewise constant on $\Vc$ is 
\begin{equation}
    |v\rangle:= \int_{\Vc(v)} dx |x\rangle.
\end{equation}
We define the subspace $\mH_{\Vc} \subset \mH_{\Ahat}$ containing the set of states which are piecewise constant on the $\epsilon$-packing  $\mV$ of $\Ac$, $|\psi\rangle_{\Ahat} = \sum_{v=1}^{|\mV|} \psi_v |v\rangle$. Similarly, $\mH_{\Vc^2}$ is the space of bipartite states of the form $\ket{\psi}_{RA} = \sum_{v=1}^{|\mV|} \psi_v \ket{v}_R \otimes \ket{v}_A.$ In particular, the subspace $\mH_{\Vc^2}\subset \mH_{R\Ahat}$ contains the maximally entangled state 
\begin{equation}\label{eq:littlephiV}
    |\phi\rangle_{R\Ahat} = \frac{1}{|\mathcal{V}|^{1/2}} \sum_{v=1}^{|\mathcal{V}|} |v_R\rangle \otimes | v_{\Ahat}\rangle
\end{equation}
Again, we will use the projector $\Pi_{\Vc^2}$ onto the set $K\subset \Ac \times \Ac$ containing all coordinates in the same $\epsilon$-ball:
\begin{equation}
    K_{\Vc^2} := \{(r,\alpha) \in \Ac \times \Ac: \hat{V}(r) = \hat{V}(\alpha) \}
\end{equation}
In this setting, we may find a non-trivial infinum by requiring $\Psi$ to have a singlet fraction $q_k(\Psi)=F$ with some $k$-dimensional singlet. Intuitively, this means that the adversary initially prepares $RA$ in an entangled state, without imposing any specifics on the basis or subspace in which that state is entangled. 

\begin{reptheorem}{thm:ent_bound}[Entanglement fraction bound]
Let $|\Phi_\epsilon\rangle$, $\mathcal{D}$, and $\mathcal{N}$ be as defined in Theorem~\ref{thm:ent_recovery}, and define $q_k(\Psi)$ to be the singlet fraction of $\Psi\in\mH_{RA}$ computed in \textit{any} $(k\times k)$-dimensional subspace of $\mH_{RA}$. Then, 
\begin{equation}
      \inf_{\substack{|\Psi\rangle_{RA}\in \mH_{RA} \\ q_{|\mathcal{V}|}(\Psi) = 1}} \sup_{\mathcal{D} \in \cptp(\mH_B, \mH_{\Ahat})} \bigl( |\mV| - \langle \Phi_\epsilon | \Pi_{\Vc^2} (\I\otimes \mathcal{D}\circ \mathcal{N} )(\dm{\Psi}) \Pi_{\Vc^2}|\Phi_\epsilon \rangle \bigr) \geq |\mV|\left(\frac{\HHH(RB)_{\rho} - 1}{\log(|\mathcal{V}|^2 - 1)}\right) 
\end{equation}
where  $\rho = (\I \otimes \mathcal{N})(\dm{\phi})$ as defined in Eq.~\ref{eq:littlephiV}.
\end{reptheorem}
\begin{proof}
    Similarly to Eq.~\ref{eq:littlephi}, we may compute 
    \begin{equation}\label{eq:little_phiV2}
        \Pi_{\Vc^2}|\Phi_\epsilon\rangle = \sum_{v \in \mV} |v\rangle_R \otimes |v\rangle_{\Ahat} := |\mV|^{1/2} |\phi\rangle_{R\Ahat}
    \end{equation}
    For any fixed $S$, since $q_k(\Psi, S) = 1$ implies $q_k(\Psi) =1$ we have the inclusion
    \begin{equation}\label{eq:lineb3_key}
        \{|\Psi\rangle \in \mH_{\Vc^2}: q_k(\Psi, S) = 1\} \subseteq \{|\Psi\rangle \in \mH_{RA}: q_k(\Psi) = 1\}.
    \end{equation}
    With this, the proof proceeds almost identically to that of Theorem~\ref{thm:ent_recovery}: 
\begin{align}
      \inf_{\substack{|\Psi\rangle \in \mH_{RA} \\ q_k(\Psi) = 1}} &\sup_{\mathcal{D} \in \cptp(\mH_B, \mH_{\Ahat})}   \langle \Phi_\epsilon | \Pi_{\Vc^2}(\I\otimes \mathcal{D}\circ \mathcal{N} )(\dm{\Psi}) \Pi_{\Vc^2} | \Phi_\epsilon\rangle  \label{line:b1}
     \\&\leq \inf_{\substack{|\psi\rangle \in \mH_{\mathbb{V}^2} \\ q_k(\psi) = 1}}  \sup_{\mathcal{D} \in \cptp(\mH_B, \mH_{\Ahat})}  \langle  \Phi_\epsilon | \Pi_{\Vc^2}(\I\otimes \mathcal{D}\circ \mathcal{N} )( \dm{\psi} ) \Pi_{\Vc^2}|\Phi_\epsilon \rangle \label{line:b2}
     \\&\leq \inf_{\substack{|\psi\rangle \in \mH_{\mathbb{V}^2} \\ q(\psi, \mH_{\Vc}) = 1}}  \sup_{\mathcal{D} \in \cptp(\mH_B, \mH_{\Ahat})}  \langle  \Phi_\epsilon |\Pi_{\Vc^2} (\I\otimes \mathcal{D}\circ \mathcal{N} )( \dm{\psi} ) \Pi_{\Vc^2}|\Phi_{\epsilon} \rangle \label{line:b3}
     \\&= \inf_{\substack{|\psi\rangle \in \mH_{\mathbb{V}^2} \\ q(\psi, \mH_{\Vc}) = 1}}  \sup_{\mathcal{D} \in \cptp(\mH_{\mathbb{V}})}  \langle  \Phi_\epsilon |\Pi_{\Vc^2} (\I\otimes \mathcal{D}\circ \mathcal{N} )( \dm{\psi} ) \Pi_{\Vc^2}|\Phi_{\epsilon} \rangle   \label{line:b4}
     \\&\leq   \sup_{\mathcal{D} \in \cptp(\mH_{\Vc})}   |\mV|\, \langle \phi | (\I\otimes \mathcal{D}\circ \mathcal{N} )  (\dm{\phi}) | \phi  \rangle \label{line:b5}
     \\&= q(R|B)_{\rho }
     \\&\leq |\mV|\left( 1 - \frac{\HHH(RB)_{\rho } - 1}{\log(k^2 - 1)} \right) \label{line:b7}
\end{align}
In Line~\ref{line:b2} we have used $\mH_{\Vc^2}\subset \mH_{RA}$, and in Line~\ref{line:b3} we have used the inclusion of Eq.~\ref{eq:lineb3_key}. For Line~\ref{line:b4} we apply Lemma~\ref{prop:cptp} in an identical manner as in the proof of Theorem~\ref{thm:ent_recovery}, and Line~\ref{line:b5} follows by specifying $|\psi\rangle = |\phi\rangle$ as defined in Eq.~\ref{eq:little_phiV2},  followed by the quantum Fano inequality (Proposition~\ref{prop:qfano}) in Line~\ref{line:b7}. Specifying $k = |\mV|$ completes the proof.

\end{proof}

\end{appendices}

\end{document}